\documentclass[10pt]{article}

\usepackage[margin=1in]{geometry}

\usepackage{amsthm}
\usepackage{amssymb}
\RequirePackage{bm}
\RequirePackage{endnotes}

\usepackage[numbers]{natbib}
\usepackage[colorlinks=true,breaklinks=true,bookmarks=true,urlcolor=blue,
 citecolor=blue,linkcolor=blue,bookmarksopen=false,draft=false]{hyperref}

\usepackage{booktabs} %
\usepackage{float}
\usepackage{wrapfig}
\usepackage[ruled,noend,linesnumbered]{algorithm2e} %

\SetAlFnt{\small}
\SetAlCapFnt{\small}
\SetAlCapNameFnt{\small}
\SetAlCapHSkip{0pt}
\SetKwProg{Function}{function}{}{}
\DontPrintSemicolon
\newcommand{\commentsymbol}{\it\color{gray}$\triangleright$~}
\newcommand{\Comment}[1]{{\commentsymbol#1}}
\SetKwFor{ForEach}{for each}{do}{endfor}
\SetFuncSty{textsc}
\newcommand{\DefineFunction}[1]{\SetKwFunction{#1}{#1}}

\usepackage{titling}
\usepackage{colortbl}
\usepackage{amsmath}
\usepackage{mathtools}
\usepackage{physics}
\usepackage{xspace}
\usepackage{enumitem}
\usepackage{tikz}
\usepackage{forest}
\usepackage{tikz-qtree}
\usepackage{linegoal}
\usepackage{cleveref}
\usepackage{multirow}
\usepackage[export]{adjustbox}
\usepackage{thm-restate}
\usepackage{bold-extra}
\usepackage{bm}
\usepackage{relsize}
\usepackage{trimclip}
\usepackage{makecell}

\usepackage{pifont}
\definecolor{darkgreen}{rgb}{0,0.5,0}
\newcommand{\cS}[1]{#1\ding{171}}
\newcommand{\cH}[1]{\textcolor{red}{#1\ding{170}}}
\newcommand{\cC}[1]{\textcolor{darkgreen}{#1\ding{168}}}
\newcommand{\cD}[1]{\textcolor{blue}{#1\ding{169}}}

\DeclareRobustCommand{\RArrow}[1]{%
\parbox{#1}{\tikz{\draw[->](0,0)--(#1,0);}}
}
\DeclareRobustCommand{\LArrow}[1]{%
\parbox{#1}{\tikz{\draw[<-](0,0)--(#1,0);}}
}
\newcommand{\rL}{{\LArrow{3pt}}}
\newcommand{\rR}{{\RArrow{3pt}}}

\tikzstyle{efce}=[fill=yellow!30!white,thick,draw=yellow!50!black,dashdotted]
\tikzstyle{efcce}=[fill=blue!15!white,draw=blue,thick,dashed]
\tikzstyle{nfcce}=[fill=red!15!white,draw=red,thick]
\tikzstyle{majortick}=[black,semithick]
\tikzstyle{majorgrid}=[very thin,gray!50]
\tikzstyle{minorgrid}=[gray,dotted]

\usepackage{autonum} %
\makeatletter
\autonum@generatePatchedReferenceCSL{Cref}
\autonum@generatePatchedReferenceCSL{cref}
\makeatother

\newcommand{\delimit}[3]{\newcommand{#1}[1]{\left#2##1\right#3}}
\delimit \ceil \lceil \rceil
\delimit \floor \lfloor \rfloor

\DeclareMathOperator*{\E}{\mathbb E}

\newcommand{\R}{\mathbb R}
\newcommand{\N}{\mathbb N}

\newcommand{\zo}{\{0,1\}}
\renewcommand{\vec}{\boldsymbol}
\newcommand{\vx}{\vec x}

\newcommand{\mediator}{\textup{\bf M}}
\newcommand{\nfcce}{\textup{NFCCE}}
\newcommand{\efcce}{\textup{EFCCE}}
\newcommand{\efce}{\textup{EFCE}}
\newcommand{\Root}{\Root}
\newcommand{\revision}[1]{{#1}}

\let\op\operatornamewithlimits
\let\eps\varepsilon
\let\mc\mathcal

\let\Root\varnothing

\newcommand{\poly}{\op{poly}}

\usepackage{color-edits}
\addauthor{a}{magenta}
\addauthor{b}{teal}   
\addauthor{G}{blue} 
\addauthor{t}{red}

\definecolor{p1color}{RGB}{31,119,180}
\definecolor{p2color}{RGB}{255,127,14}
\definecolor{p3color}{RGB}{44,160,44}
\definecolor{p4color}{RGB}{214,39,40}

\newcommand{\pone}{{\ensuremath{\color{p1color}\blacktriangle}}\xspace}
\newcommand{\ptwo}{{\ensuremath{\color{p2color}\blacktriangledown}}\xspace}

\newcommand*\circled[1]{\tikz[baseline=(char.base)]{
            \node[shape=circle,draw,inner sep=1pt] (char) {\normalfont\scriptsize\sffamily #1};}}
            
\newcommand{\defeq}{\mathrel{:\mkern-0.25mu=}}

\newcommand{\rele}{\bowtie}

\newcommand{\sr}{\texttt{sr}}
\newcommand{\curxi}{{\vec{\xi}}^{(T)}}
\NewDocumentCommand{\Sr}{O{i}}{\Xi^{\sr}_{#1}}
\DeclareMathOperator{\co}{co}
\NewDocumentCommand{\margi}{O{i}}{\vec{m}_{#1}}
\NewDocumentCommand{\marg}{O{i}}{\vec{m}_{#1}}
\NewDocumentCommand{\tildemarg}{O{i}}{\vec{m}_{#1}^{(t)}}
\NewDocumentCommand{\hatmarg}{O{i}}{\vec{m}_{\hat{\vec\xi},#1}}
\newcommand{\nodefont}{\bfseries\sffamily}
\newcommand{\pfont}[2]{{\nodefont\color{p#1color} #2}}

\newcommand{\Vrs}{V^{\texttt{rs}}}
\newcommand{\Ers}{E^{\texttt{rs}}}
\newcommand{\Grs}{G^{\texttt{rs}}}

\usepackage{mleftright}

\newcommand{\X}{\mathcal{X}}

\newcommand{\centercell}[1]{\multicolumn{1}{c}{#1}}

\newcommand{\payoffspace}[3]{%
\begin{tikzpicture}
  \node at (0,0) {\includegraphics[scale=.7]{plots/#1#2#3_payoff_space.pdf}};
  \ifthenelse{#2=3}{
    \node[anchor=north west,fill=white,rounded corners=2pt] at (-2.3,2.7) {\scalebox{.8}{Game: \gamelbl{$^{#2}$#1#3}}};
  }{
    \node[anchor=south,fill=white,rounded corners=2pt,outer ysep=1mm] at (.5,2.35) {\scalebox{.8}{Game: \gamelbl{$^{#2}$#1#3}}};
  }
\end{tikzpicture}
}%

\DeclareRobustCommand\gamelbl[1]{\tikz[baseline,anchor=base]  \node[fill=black!50,rounded corners=2pt,inner xsep=.9mm,inner ysep=.7mm,text=white] (X) {\scalebox{1}{\sffamily #1}};}

\newcommand{\makelegend}{
    \begin{tikzpicture}
        \tikzstyle{efce}=[fill=yellow!30!white,thick,draw=yellow!50!black,dashdotted]
        \tikzstyle{efcce}=[fill=blue!15!white,draw=blue,thick,dashed]
        \tikzstyle{nfcce}=[fill=red!15!white,draw=red,thick]
        \tikzstyle{lbl}=[inner sep=0mm]
    
        \begin{scope}
            \filldraw[nfcce] (0, 0) rectangle +(.7,.25);
            \node[lbl,anchor=south west] at (.85, 0) {\small NFCCE};
        \end{scope}
        \begin{scope}[xshift=3cm]
            \filldraw[efcce] (0, 0) rectangle +(.7,.25);
            \node[lbl,anchor=south west] at (.85, 0) {\small EFCCE};
        \end{scope}
        \begin{scope}[xshift=6cm]
            \filldraw[efce] (0, 0) rectangle +(.7,.25);
            \node[lbl,anchor=south west] at (.85, 0) {\small EFCE};
        \end{scope}
    \end{tikzpicture}
}

\theoremstyle{plain}
\newtheorem{theorem}{Theorem}[section]

\newtheorem{lemma}[theorem]{Lemma}

\theoremstyle{definition}
\newtheorem{definition}[theorem]{Definition}

\newtheorem{example}[theorem]{Example}
\theoremstyle{remark}
\newtheorem{remark}[theorem]{Remark}
\numberwithin{theorem}{section}

\pgfdeclarelayer{background}
\pgfsetlayers{background,main}

\usetikzlibrary{positioning,arrows}

\newskip\algoinsideskipamount \algoinsideskipamount=2pt plus 0pt minus 0pt

\SetAlgoInsideSkip{algoinsideskip}
\SetAlCapHSkip{1mm}
\SetAlCapSty{kwsty}
\SetKwSty{kwsty}
\SetFuncSty{fnm}
\SetArgSty{upshape}
\SetFuncArgSty{upshape}
\SetKwComment{Hline}{}{\vspace{-3mm}\textcolor{gray}{\hrule}\vspace{1mm}}
\definecolor{darkgrey}{gray}{0.3}
\definecolor{commentcolor}{gray}{0.5}
\SetKwComment{Comment}{\color{commentcolor}[$\triangleright$\ }{}
\SetCommentSty{}
\SetNlSty{}{\color{black!50}\sffamily}{}

\SetKwProg{XFn}{function}{}{}
\SetKwProg{XSubr}{subroutine}{}{}

\SetKwFor{ForEach}{for~each}{do}{end}
\crefalias{AlgoLine}{line}%

\title{Optimal Correlated Equilibria in General-Sum Extensive-Form Games: Fixed-Parameter Algorithms, Hardness, and Two-Sided Column-Generation}

\author{
Brian Hu Zhang \\ \small Carnegie Mellon University 
\and Gabriele Farina \\ \small MIT 
\and Andrea Celli \\ \small Bocconi University 
\and Tuomas Sandholm \\ {\small \begin{tabular}{c}Carnegie Mellon University\\ \small Strategy Robot, Inc.\\ \small Strategic Machine, Inc\\ \small Optimized Markets, Inc.\end{tabular}} }

\begin{document}
\maketitle

\begin{abstract}%
We study the problem of finding {\em optimal correlated equilibria} of various sorts in extensive-form games: \textit{normal-form coarse correlated equilibrium (NFCCE)}, \textit{extensive-form coarse correlated equilibrium (EFCCE)}, and \textit{extensive-form correlated equilibrium (EFCE)}. 
We make two primary contributions. 
First, we introduce a new algorithm for computing optimal equilibria in all three notions. Its runtime depends exponentially only on a parameter related to the information structure of the game. 
We also prove a fundamental complexity gap: while our size bounds for NFCCE are similar to those achieved in the case of team games by \citet{Zhang22:Team_DAG}, this is impossible to achieve for the other two concepts under standard complexity assumptions. 
Second, we propose a {\em two-sided column generation approach} for use when the runtime or memory usage of the previous algorithm is prohibitive. Our algorithm improves upon the {\em one-sided} approach of \citet{Farina21:Connecting} by means of a new decomposition of correlated strategies which allows players to re-optimize their sequence-form strategies with respect to correlation plans which were previously added to the support. 
Experiments show that our techniques outperform the prior state of the art for computing optimal general-sum correlated equilibria. %
\end{abstract}

\section{Introduction}
Recent algorithms for computing {\em Nash equilibria} in zero-sum imperfect-information extensive-form games have led to breakthroughs, most notably strong agents for two-player no-limit Texas hold'em poker~\cite{Moravvcik17:DeepStack,Brown18:Superhuman}. However, in general-sum and/or multi-player games, computing Nash equilibria is hard even in normal-form games~\cite{Chen09:Settling}. Further, in real-world situations, the assumption of player strategies being {\em independent}, as posited in Nash equilibrium, might not hold true. For example, agents may share conventions, or communicate with a trusted mediator. Both of these concerns motivate the definition and computational study of notions of {\em correlated equilibria}.

In correlated equilibria, an outside {\em mediator} can recommend, but not enforce,  certain actions. More precisely, the mediator first draws a strategy {\em profile} from a publicly-agreed distribution, and recommends to each player their chosen strategy. The players may then choose whether to accept the recommendation or to {\em deviate} and play an arbitrary action instead. A {\em normal-form correlated equilibrium} (NFCE)~\cite{Aumann74:Subjectivity} is a distribution of profiles for which no player is ever incentivized to deviate. In a {\em normal-form coarse correlated equilibrium} (NFCCE)~\cite{Moulin78:Strategically,celli2019computing}, each player must choose to commit to following the recommendation {\em before} receiving it---if a player commits, she must play the recommended strategy; if she does not commit, she does not receive a recommendation.

Both above notions of correlated equilibria were originally defined only for normal-form games. More recently, \citet{Stengel08:Extensive}, and \citet{Farina20:Coarse} defined and studied notions of correlated equilibria in extensive-form games. In an {\em extensive-form correlated equilibrium} (EFCE), each player receives recommendations throughout the game at each of their decision point, and again can choose to follow or ignore the recommendation. In an {\em extensive-form coarse correlated equilibrium} (EFCCE), at each decision point, each player must commit to following the recommendation before seeing it. In both cases, a player that deviates no longer receives recommendations for the remainder of the game. We refer the reader to \Cref{fig:difference between correlated solution concepts} for a visual summary of the difference between the solution concepts.

Our focus is on computing {\em optimal} NFCCEs, EFCCEs, and EFCEs, which are the equilibria that maximize a given linear objective function. Computing optimal correlated equilibria, in any of these notions, is NP-hard in the size of the game tree, even in two-player games with chance nodes, or three-player games without chance nodes~\citep{Stengel08:Extensive}. Some special cases are known to be solvable efficiently. \citet{Stengel08:Extensive} show that in two-player games without chance moves, optimal equilibria in all three equilibrium notions can be computed in polynomial time. More recently, \citet{Farina20:Polynomial} extend the positive result to so-called {\em triangle-free games}, which strictly include all two-player games with public chance actions.

The problem of computing \emph{one} EFCE (and, therefore, one NFCCE/EFCCE) can be solved in polynomial time in the size of the game tree~\cite{huang2008computing} via a variation of the {\em Ellipsoid Against Hope} algorithm~\cite{papadimitriou2008,jiang2015polynomial}. Moreover, there exist decentralized no-regret learning dynamics guaranteeing that the empirical frequency of play after $T$ rounds is an $O(1/\sqrt{T})$-approximate EFCE with high probability, and an EFCE almost surely in the limit~\cite{celli2020no,farina2021simple}. Using regret minimizers to play large multi-player games has already led to superhuman practical performance in multi-player poker~\cite{Brown19:Superhuman}. As stated above, however, computing optimal equilibria is much harder.

Correlated equilibria have a close relationship with {\em adversarial team games}, that is, games where two teams compete against each other \cite{vonStengel97:team,celli18:computational}. An efficient algorithm for representing the space of correlated strategies of a team of players also gives an efficient algorithm for solving adversarial team games. Until recently, the state of the art for solving team games was to represent the space of correlated strategies of the team, as if to compute an extensive-form correlated equilibrium of that team~\cite{Farina21:Connecting}. Recently, \citet{Zhang22:Team_TreeDecomp} and \citet{Zhang22:Team_DAG} have developed new methods of solving team games based on {\em public states}. Their work gives a construction of the decision space of a team whose complexity is dependent on natural parameters of the game. However, their construction does not extend to general-sum correlation:  for that, we are not only interested in the player reach probabilities of the terminal states, but also, among other things, in the {\em marginal} strategies of each individual player. This difference, as we will explain, creates a \emph{critical separation between adversarial team games and general-sum correlation}.

\paragraph{Contributions and paper structure}

This paper makes a number of contributions related to the computation of \emph{optimal} (\textit{i.e.}, one that maximizes a given linear objective function, such as social welfare or any weighted sum of expected player utilities) NFCCE, EFCCE, and EFCE in general multi-player general-sum extensive-form games. At a high level, we distinguish between \emph{conceptual}, \emph{complexity-theoretic}, and \emph{algorithmic} contributions.
\begin{itemize}
\item[-] \emph{Conceptual contributions.} At the conceptual level, we show that the problem of computing an optimal NFCCE, EFCCE, and EFCE, can be converted into the problem of computing an optimal strategy for a player in a suitably-constructed  game. The equivalent game, which we call a \emph{mediator-augmented game}, explicitly captures the decision problem that each player would face if the correlation device were an explicit player in the game, called the \emph{mediator}. The action space of the mediator depends on the solution concept being analyzed: NFCCE, EFCCE, or EFCE. 

While the mediator-augmented formalism greatly simplifies the treatment---providing what we hope will be an important conceptual framework for further analysis of these solution concepts---this game reformulation preserves the computational aspects of computing an optimal equilibrium, including their hardness aspects. Indeed, a key point regarding the mediator-augmented game is that the mediator faces \textit{imperfect recall}. This is because the mediator cannot leak information across the players, so the mediator has to forget what it has observed about the other players when making a recommendation to a given player. Otherwise, the mediator's recommendations would not form a correlated profile at all, much less any equilibrium.  

Optimizing for the strategy of an imperfect-recall player (here, the mediator) is known to be hard~\cite{Koller92:Complexity,Chu01:NP}. To tackle the issue, in our paper we study effective extended formulations (in the mathematical programming sense, \textit{e.g.}, \cite{Conforti10:Extended}) of the decision space of the mediator, by removing the imperfect recall at the expense of a (worst-case exponential) increase in the number of decision points for the mediator player. 

\item[-] \emph{Complexity-theoretic contributions.} We then proceed to show how certain recent results regarding parameterized complexity of imperfect-recall decision problems can be applied to the mediator-augmented game. A critical technical step in applying those results lies in characterizing the complexity of the \emph{public states} of the decision problem faced by the mediator in the mediator-augmented game as a function of the original (not mediator-augmented) input game. Specifically, we give bounds on the size of the public states of mediator-augmented games for each of the solution concepts as a function of the depth $d$, the maximum branching factor $b$, and a suitably-defined \emph{information-complexity} $k$ of the input game that is independent of the solution concept. %
However, our overall complexity bounds are different depending on the solution concept: the bound for NFCCE in particular does not depend exponentially on the depth of the game, whereas the bounds for EFCCE and EFCE do. We show that this difference is inherent, therefore contributing new complexity-theoretic separations between the solution concepts.
\begin{itemize}
    \item[i.] We show that an optimal EFCE in an extensive-form game can be computed by solving a linear program of size $O^*((bd)^k)$, where the notation $O^*$ suppresses factors polynomial in the size of the game (\Cref{th:dag size}). For optimal EFCCE and optimal NFCCE, we establish bounds of $O^*((b+d-1)^k)$ and $O^*((b+1)^k)$, respectively.
    \item[ii.] In games with {\em public player actions}, we show that the bounds for NFCCE and EFCCE can be further improved to $O^*(3^k)$ and $O^*(d^k)$, respectively (\Cref{th:public actions}). We show that the bound for EFCE \textit{cannot} be improved in this manner.
    \item[iii.] In {\em two-player} games with {\em public chance actions}, our algorithm runs in polynomial time (\Cref{th:pubchance}) for all three solution concepts. The problem in this setting had already been shown to be solvable in polynomial time using a different technique by \citet{Farina20:Polynomial}; we match their results and discuss the relationship between our algorithm and theirs in \Cref{se:pubchance-discussion}.
    \item[iv.] We show that the gap between the NFCCE bound and the EFCCE and EFCE bounds is fundamental. Matching the bound for NFCCE---in particular, removing the dependence on $d$---is impossible for EFCCE and EFCE under standard complexity assumptions, demonstrating a \emph{fundamental complexity-theoretic gap} for coarse correlation between normal and extensive form (\Cref{th:w1-hardness}).
\end{itemize}
\revision{When $k$ is a constant, our algorithms are, to our knowledge, the first efficient algorithms for the problem of computing an optimal correlated equilibrium in any of the three solution concepts. A comparison between our complexity results and those of past papers~\cite{Stengel08:Extensive,Farina20:Polynomial} can be found in \Cref{tab:summary}.}

\item[-] \emph{Algorithmic contributions.} We propose two main algorithms for computing optimal correlated equilibria in all three solution concepts.
\begin{itemize}
\item[i.] We operationalize the positive complexity results established above (\Cref{th:dag size,th:public actions,th:pubchance}) via \Cref{al:dag}. It computes an optimal strategy for the mediator in the mediator-augmented game via linear programming. At its core, 
the algorithm is based on the idea that the imperfect-recall strategy space of the mediator is the projection of the set of flows in a suitable high-dimensional directed acyclic graph (DAG), called  the {\em team belief DAG}~\cite{Zhang22:Team_DAG}. 
To our knowledge, this characterization of the complicated polytope of feasible correlated equilibria as the projection of a simpler set of flows in a higher dimension is the first example of an extended formulation (in the mathematical programming sense, \textit{e.g.}, \cite{Conforti10:Extended}) for these solution concepts.

One cannot directly apply the fixed-parameter results of \citet{Zhang22:Team_DAG}, as that would result in a worse bound. Instead, the above results are proven by carefully analyzing the size of the resulting construction with the special structure of the mediator-augmented games in mind.

\item[ii.] We propose a new practical approach to computing optimal correlated equilibria which we call \emph{two-sided column generation} (\Cref{se:colgen}).  We start by deriving an LP formulation based on the strategy polytope of~\citet{Stengel08:Extensive} and on the notion of \emph{semi-randomized correlation plan} introduced by~\citet{Farina21:Connecting} in the context of team games.
In the latter of those two prior approaches, one player is chosen to play a normal-form strategy and the other plays a mixed (sequence-form) strategy. Our approach improves upon this by allowing the master LP to {\em select} which player is chosen to play the mixed strategy, thereby increasing the space of correlation plans that can be represented for any given support, and leading to a tighter master problem. In practice, we find that this change yields a speed improvement over the algorithm of \citet{Farina21:Connecting} in almost all of the games tested, and this speed improvement can be greater than two orders of magnitude.

\end{itemize}

\noindent Our two solving techniques are complementary: where the parameter $k$ is small, writing out the DAG is superior; where it is large, the two-sided column generation is faster and more frugal in its memory usage. Furthermore, the value of $k$ can be easily computed, enabling an efficient choice between these two approaches. In experiments (\Cref{se:experiments}), we demonstrate state-of-the-art practical performance compared to prior state-of-the-art techniques with at least one, and sometimes both, of our techniques.  We also introduce two new benchmark games: a 2-vs-1 adversarial team game we call the {\em tricks game} which is the trick-taking (endgame) phase of the card game bridge, and the {\em ride-sharing game} in which two drivers seek to earn points by serving requests across a road network modeled as an undirected graph. In the tricks game, we demonstrate empirically that, even for small endgames with only three cards per player remaining, relaxing the game to be perfect information---as so-called {\em double dummy} bridge endgame solvers do (\textit{e.g.},~\citep{ginsberg1999gib})---causes incorrect solutions and game values to be generated, demonstrating the need for imperfect-information game analysis.
\end{itemize}
\begin{table}
\centering\smaller
    \begin{tabular}{p{4.5cm}cccc}
    \toprule
        & \multicolumn{4}{c}{\bf Game class}\\
        \textbf{Algorithm} & No chance & Public chance & Triangle-free & Information complexity $k$ \\
         \midrule
 \citet{Stengel08:Extensive} & poly & --- & --- & ---\\
    \rowcolor{gray!20}
       \citet{Farina20:Polynomial} & poly & poly & poly & --- \\
       \makecell[l]{Correlation DAG\\{}[\textbf{this paper}]} & \makecell{poly \\ (\Cref{th:pubchance})} & \makecell{poly \\ (\Cref{th:pubchance})} & \makecell{exp \\~} & \makecell[c]{$O^*((bd)^k)$ \\ (\Cref{th:dag size})} \\
       \bottomrule
    \end{tabular}
    \caption{
\revision{Comparison of results in our paper with previously-known results about the computation of optimal correlated equilibria. `exp' and `poly' mean exponential time and polynomial time, respectively. `---' means that the analysis of that paper cannot handle that class of games. All no-chance games are public-chance (trivially), and all public-chance games are triangle-free~\cite{Farina20:Polynomial}. Both these inclusions are strict. The correlation DAG algorithm requires the game to be timeable. The $O^*((bd)^k)$ result is for EFCE; the bound is better for EFCCE and NFCCE. Column generation has poor theoretical guarantees but can work well in practice compared to the correlation DAG, especially when the information complexity, $k$, is large. }}\label{tab:summary}
\end{table}

\section{Preliminaries}\label{se:prelim}
In this section, we review common notions for correlation in extensive-form games.

\subsection{Extensive-Form Games}

We start with the definition of (imperfect-information) extensive-form games, that is, tree-form games in which players might not observe all actions.

\begin{definition}
    An {\em extensive-form game} $\Gamma$ with $n$ {\em players}, which we will identify with the positive integers $[n] = \{1, \dots, n\}$ consists of the following:
    \begin{enumerate}
        \item A rooted tree of {\em nodes} $\mc H$, where the edges are labelled with {\em actions}.  The root node of $\mc H$ will be denoted $\Root$. The set of leaves, or {\em terminal nodes} in $\mc H$ will be denoted $\mc Z$. The set of actions at a node $h \in \mc H$ will be denoted $A_h$. The child reached by following action $a$ at node $h$ will be denoted $ha$.
        \item A partition $\mc H_0, \mc H_1, \dots, \mc H_n$ of the set of nonterminal nodes, where $\mc H_i$ for $i > 0$ is the set of decision nodes of player $i$ and nodes in $\mc H_0$ are chance nodes.
        \item For each player $i \in [n]$, a partition $\mc I_i$ of $\mc H_i$ into {\em information sets}, also known as {\em infosets} for short. The set of actions at every node in a given infoset $I$ must be the same, and we will denote it $A_I$.
        \item For each player $i \in [n]$, a {\em utility vector} $\vec u_i \in \R^{\mc Z}$, where $u_i[z]$ is the utility that player $i$ achieves upon reaching terminal node $z$.
        \item For each chance node $h \in \mc H_0$, a fixed distribution $p(\cdot \mid h)$ over $A_h$. We will use $p(z)$ to denote the probability that chance plays all actions on the path from root to $z$.
    \end{enumerate}
\end{definition}

Information sets $I \in \mc I_i$ contain all those nodes that Player $i$ cannot distinguish among when acting at those nodes. This is further elucidated in the following example.

\begin{example}\label{example:efg}
    As an example, consider the example game of \Cref{fig:example-game}. The game has two players ($n=2$), whose nodes are pictorially marked with \pone for Player 1 and \ptwo for Player 2 respectively, and 19 nodes (denoted {\nodefont a} through {\nodefont s}), of which nine ({\nodefont a} through {\nodefont i}) are nonterminal. The root node is a chance node, at which the chance player moves uniformly at random. Being the only chance node, it follows that $\mc H_0 = \{\text{\nodefont a}\}$. Player 1 (\pone) observes the outcome of the chance node, and can pick between a left or a right action. Player 2 (\ptwo) however does \emph{not} observe the outcome of the chance node; rather, the player only observes the choice of Player 1. This imperfect knowledge of the state is encoded by the information partition $\mc I_2$ of Player 2, which contains the two information sets $\{\{\text{\pfont2 d}, \text{\pfont2 e}\}, \{\text{\pfont2 f}, \text{\pfont2 g}\}\}$, denoted in the figure with dotted lines connecting the nodes in the same information set. If the game hits state {\pfont2 d}, then Player 1 (\pone) gets to play a second move. However, Player 1 will not observe the action chosen by Player 2 at {\pfont2 d}; this is captured again by the information set $\{\text{\pfont1 h},\text{\pfont1 i}\}$. Nodes \pfont1 b and \pfont1 c do not bear any uncertainty, and are therefore singleton elements in their corresponding information sets. In summary, the information partitions of the players are $\mc I_1 = \{\{\text{\pfont1 b}\}, \{\text{\pfont1 c}\}, \{\text{\pfont1 h},\text{\pfont1 i}\}\}$ and $\mc I_2 = \{\{\text{\pfont2 d}, \text{\pfont2 e}\}, \{\text{\pfont2 f}, \text{\pfont2 g}\}\}$.
    At terminal nodes, the payoffs for \pone, \ptwo are listed below the node. \ptwo has utility zero at every terminal node. Examples of correlated equilibria for this game are given in \Cref{se:example}.
\end{example}
\begin{figure}[t]
    \centering
    
\def\figscale{.9}

\tikzset{
    every node/.style={draw, text height=1ex, text depth=0pt, inner ysep=6pt, outer sep=0pt, font=\nodefont, align=center},
    p1/.style={
            regular polygon,
            regular polygon sides=3,
            rounded corners=1.2,
            inner sep=2pt, fill=p1color, draw=none, text=white},
    p2/.style={p1, shape border rotate=180, fill=p2color},
    med/.style={circle,inner sep=3pt, fill=p4color, draw=none},
    infoset1/.style={-, densely dotted, ultra thick, color=p1color},
    infoset/.style={infoset1, color=p2color},
    util/.style={},
    terminal/.style={text height=2.5ex, draw=none},
    played/.style={ultra thick, p3color},
    iiname/.style={midway,draw=none,p#1color,fill=white}
}

\forestset{
default preamble={for tree={s sep=.5cm}},
parent/.style={no edge,tikz={\draw (#1.parent anchor) to (!.child anchor);}},
}

\newcommand{\utilformat}[1]{\textbf{\textsf{#1}}}
\newcommand{\util}[2]{{\begin{tabular}{c}{#1}\\{\textcolor{p1color}{#2},\textcolor{p2color}{0}}\end{tabular}}}
    \scalebox{\figscale}{
\begin{forest}
[a,tier=1,no edge
    [b,tier=2,p1,name=b
        [d,tier=3,name=d,p2
            [h,p1,tier=4,name=h
                [\util p2,tier=5,terminal]
                [\util q1,tier=5,terminal]
            ] 
            [i,p1,tier=4,name=i
                [\util r0,tier=5,terminal]
                [\util s1,tier=5,terminal]
            ] 
        ]
        [e,tier=3,name=e,parent=c,p2
            [\util j0,tier=4,terminal]
            [\util k1,tier=4,terminal]
        ]
    ]
    [c,tier=2,p1,name=c
        [f,tier=3,name=f,parent=b,p2
            [\util l1,tier=4,terminal]
            [\util m0,tier=4,terminal]
        ]
        [g,tier=3,name=g,p2
            [\util n0,tier=4,terminal]
            [\util o1,tier=4,terminal]
        ]
    ]
]
\draw[infoset] (d) -- (e);
\draw[infoset] (f) -- (g);
\draw[infoset,p1color] (h) -- (i);
\end{forest}
}
\vspace{1ex}
    \caption{An example game, between two players \pone (P1) and \ptwo (P2). The root node is a chance node, at which chance moves uniformly at random. Dotted lines connect nodes in the same information set. Bold lowercase letters are the names of nodes. We will refer to infosets by naming all the nodes within them; for example, \pfont1{b} and \pfont2{de} are infosets. At terminal nodes, the utility of \pone is listed below the name of the node. \ptwo has utility zero at every terminal node, and in this game the only role of \ptwo is to incentivize \pone to act in a certain way.}
    \label{fig:example-game}
\end{figure}

We will use $\preceq$ to denote the precedence relation induced by a tree. For example, $h \preceq h'$ if $h$ is an ancestor of $h'$ in the tree. If $S$ and $S'$ are sets of nodes, we will use $S \preceq h$ or $S \succeq h$ to mean that there exists $s \in S$ for which $s \preceq h$ or $s\succeq h$ (respectively), and $S \preceq S'$ to mean that there exist $h \in S$ and $h' \in S'$ with $h \preceq h'$. We will use $h \land h'$ to denote the lowest common ancestor of $h$ and $h'$.

The {\em sequence} $\sigma_i(h)$ of player $i$ at node $h$ are the sequence of information sets reached and actions played by $i$ on the root $\to h$ path, {\em not} including the infoset at $h$ itself even when $h$ is a decision node of player $i$. We assume that every player has {\em perfect recall}---that is, at every player $i$ infoset $I$, every $h \in I$ has the same sequence, denoted $\sigma_i(I)$. The set of sequences of player $i$ will be denoted $\Sigma_i := \{ \sigma_i(h) : h \in \mc H \}$. The {\em empty sequence} of player $i$, $\sigma_i(\Root)$, will be denoted $\Root_i$.

In perfect-recall games, a sequence can be identified with infoset-action pair $Ia$. We will use this identification, and moreover, we will identify $Ia$ with the set of nodes $Ia := \{ ha : h \in I \}$. This will allow us to make use of statements such as ``the nodes of $Ia$ are one level deeper than those of $I$''.

A {\em pure strategy} for a player $i$ is an assignment of one action to each information set $I \in \mc I_i$. %
The {\em sequence form representation} of a pure strategy is the vector $\vec x_i \in \zo^{\Sigma_i}$, where $\vec x_i[\sigma_i] = 1$ if player $i$ plays every action on the path from $\Root_i$ to $\sigma_i$.  We will use $\Pi_i$ to denote the set of all sequence-form pure strategies.

For infosets or nodes $v$, we will use $\vec x_i[v]$ as overloaded notation for $\vec x_i[\sigma_i(v)]$. If $\vec x_i[v] = 1$, we say that $\vec x_i$ {\em plays to} $v$. A {\em mixed strategy}, also denoted $\vec x_i$, is a distribution over pure strategies. The sequence form of a mixed strategy is the appropriate convex combination of sequence forms of pure strategies. The set of mixed strategies of player $i$ is denoted by $\X_i = \co \Pi_i$. For {\em perfect-recall games}, $\X_i$ is a convex polytope characterized by a linear constraint system of size $O(\abs{\Sigma_i})$, containing one constraint for each information set~\cite{Romanovskii62:Reduction,Koller94:Fast,Stengel96:Efficient}. Note that the sequence form is also well-defined, and is still a convex polytope, even for {\em imperfect-recall} players; however, in this setting, unless $\mathsf{P = NP}$, the smallest constraint system defining the polytope may be exponential~\cite{Koller92:Complexity}.

A {\em pure profile} $\vx = (\vx_1, \dots, \vx_n)$ is a collection of pure strategies, one per player. For a pure profile $\vx$, we define $\vx[h] := \prod_{i \in N} \vx_i[h] \in \zo$ to be the indicator that all players  play all actions on root $\to h$ path.
The {\em expected utility} of player $i$ under $\vx$  is $u_i(\vx) := \E_{z \sim \vec x} u_i[z]$, where $z \sim \vx$ denotes sampling a terminal node $z$ by following the profile $\vec x$, where, since $\vec x$ is pure, the expectation is over nature's actions. A {\em correlated profile} $\mu$ is a distribution over pure profiles. %

The $\ell$th {\em layer} of the game tree consists of all nodes exactly distance $\ell$ from $\Root$. That is, layer $0$ contains only the root, the layer $1$ contains all children of the root, and so on. The {\em depth} of the game tree is the largest $\ell$ for which layer $\ell$ is nonempty. We will call an extensive-form game {\em timeable} if no infoset contains nodes in multiple layers. This is a fairly mild assumption commonly used in the extensive-form game literature (see, \textit{e.g.}, \cite{Jakobsen16:Timeability} for a discussion) and (implicitly) universal in the reinforcement learning literature. Throughout this paper, unless otherwise stated, we consider only timeable games. 

A table summarizing the notation used in this paper can be found in the appendix (\Cref{tab:notation}).

\subsection{Correlated Equilibria in Games}

Most notions of correlated equilibria in extensive-form games, including {\em normal-form coarse correlated equilibrium} (NFCCE), {\em extensive-form coarse correlated equilibrium} (EFCCE), and {\em extensive-form correlated equilibrium} (EFCE), can be thought of as correlated strategies of play that can be {\em enforced by a mediator}. The mediator first computes and publicly announces a correlated profile $\mu$. Then, {\em privately}, the mediator selects a pure profile $\vx \sim \mu$. Then, whenever a player $i$ reaches an infoset $I$, the mediator gives a {\em recommendation} that $i$ play \revision{the action played by $\vx_i$ at $I$}. The player may also choose to {\em deviate}, in which case they do not need to follow the recommendations of the mediator, but the mediator also no longer {\em gives} recommendations for the remainder of the game. The different notions of correlation are separated by what types of deviations are allowed (see also \Cref{fig:difference between correlated solution concepts}).
\begin{itemize}
    \item In NFCCE, a player may only deviate at the very beginning of the game. If she chooses not to deviate, she must follow all mediator recommendations for the whole game.
    \item In EFCCE, a player may deviate at each of her infosets {\em before} seeing a recommendation. However, if she chooses not to deviate, she must play the recommended action.
    \item In EFCE, a player may deviate at each of her infosets {\em after} seeing a recommendation, by instead playing a different action.
\end{itemize}

\begin{figure}[htp!]
\centering\scalebox{0.92}{
    \begin{tikzpicture}[xscale=4,yscale=2.2]
      \draw[very thick,->] (0, 0) --node[xshift=-2.99cm,yshift=-1mm,rotate=90]{\textbf{What is revealed?}} (0, 1.4);
      \draw[very thick,->] (0, 0) --node[yshift=-1.2cm]{\textbf{When is the commitment to follow made?}} (2.4, 0);
      \draw[gray,dashed, thin] (0, .7) -- (2.25, .7);
      \draw[gray,dashed, thin] (1.125, 0) -- (1.125, 1.4);

      \node[anchor=north] at (.5, 0) {\begin{minipage}{4.5cm}\centering\small Before seeing the\\ recommendation\end{minipage}};
      \node[anchor=north] at (1.75, 0) {\begin{minipage}{4.5cm}\centering\small After seeing the\\ recommendation\end{minipage}};

      \node[anchor=east] at (0, .4) {\begin{minipage}{2.4cm}\small\centering Single move\\incrementally\end{minipage}};
      \node[anchor=east] at (0, 1) {\begin{minipage}{2.4cm}\small\centering Whole strategy\\upfront\end{minipage}};

      \node[] at (0.5625, .35) {\begin{minipage}{4cm}
                                  \centering \textbf{EFCCE}\\
                                  \centering\small{\citet{Farina20:Coarse}}
                                \end{minipage}};
      \node[] at (1.8, .35) {\begin{minipage}{4.5cm}
                                  \centering \textbf{EFCE}\\
                                  \centering\small{\citet{Stengel08:Extensive}}
                                \end{minipage}};
      \node[] at (0.5625, 1.05) {\begin{minipage}{3.6cm}
                                  \centering \textbf{NFCCE}\\
                                  \centering \small{\citet{Moulin78:Strategically}}
                                \end{minipage}};
      \node[] at (1.8, 1.05) {\begin{minipage}{4.2cm}
                                  \centering \textbf{NFCE}\\
                                  \small{\citet{Aumann74:Subjectivity}}
                                \end{minipage}};
    \end{tikzpicture}}
    \caption{%
    Comparison of different notions of correlation in extensive-form games.
    }
    \label{fig:difference between correlated solution concepts}
\end{figure}

The fourth notion of equilibrium, called {\em  normal-form correlated equilibrium} (NFCE), is often known as simply the {\em correlated equilibrium}. In NFCE, the mediator tells each player her entire pure strategy $x_i$ at the start of the game, at which point the player may choose to deviate. It is known computing optimal NFCEs is NP-hard even in two-player games without chance nodes (unlike for the three notions we study in this paper)~\cite{Stengel08:Extensive}, making it a distinctly difficult problem that is out of the scope of this paper. Thus, throughout this paper, we use {\em ``correlated equilibrium''} to generically refer to any of the three notions of correlated equilibrium that we investigate.

\paragraph{Triggers} To formalize these notions, we use the language of {\em deviations} introduced by \citet{Gordon08:No}.
Each deviation consists of a \emph{trigger} and a \emph{continuation strategy}, which specifies the behaviour of the player when they decide to deviate from the mediator's recommendation. The trigger determines the point of the game in which the deviating player stops following the recommendation to start playing as prescribed by the continuation strategy. 
Each of the solution concepts that we consider has a different set of triggers. In an NFCCE each player is allowed to deviate only at the beginning of the interaction, before any recommendation is observed. Therefore, each player $i$ will have the empty sequence $\Root_i$ as their trigger. In an EFCCE triggers are the information sets of the game, while in an EFCE players may get triggered after observing a specific action recommendation at a specific information set of the game. 

\begin{definition}
    A {\em trigger} $\tau$ is:
    \begin{itemize}
        \item for NFCCE, the empty sequence $\Root_i$ for some player $i \in [n]$;
        \item for EFCCE, an infoset; and
        \item for EFCE, a sequence.
    \end{itemize}
\end{definition}

Given a solution concept $c \in \{\nfcce, \efcce, \efce\}$, we denote by $\mc T^c$ the set of all triggers for that concept, and $\mc T^c_i$ the set of all triggers of player $i$. %
Given a trigger $\tau \in \mc T^c$, we use $\bar\tau$ to denote where $\tau$ can be activated. That is, $\bar\tau = I$ if $\tau = Ia$ is a non-root sequence, or else $\bar\tau = \tau$. We must make this distinction because EFCE triggers are activated not by reaching a part of a game tree, but by receiving a recommendation $a$ {\em after} reaching  a part of the game tree. We use $\Sigma^{\bar\tau}_i$ to denote the set of all sequences $\sigma \succeq \bar\tau$ of player $i$.

A (pure) {\em continuation} $\vx_i' \in \{0, 1\}^{\Sigma^{\bar\tau}_i}$ following a trigger $\tau$ of player $i$ is a pure strategy defined on all infosets $I \succeq \bar\tau$. In sequence form, $\vx_i'$ is indexed by sequences $\sigma \succeq \bar\tau$, and $\vx_i'[\sigma]=1$ if the player plays all actions on the path from $\bar\tau$ to $\sigma$. Mixed continuation strategies are defined analogously. %

\paragraph{Deviations} A pair $(\tau, \vx_i')$, consisting of a trigger $\tau$ of player $i$ and a pure continuation $\vx_i'$ following $\bar\tau$, defines a {\em deviation} $\phi^{(\tau, \vx_i')} : \Pi_i \to \Pi_i$ in the following manner: $\phi^{(\tau, \vx_i')}(\vx)$ is the pure strategy that plays according to the original strategy $\vx_i$ unless it prescribes $\tau$, in which case it replaces it strategy with the continuation $\vx_i'$ wherever the latter is defined. Formally,
\begin{align}
    \phi^{(\tau, \vx_i')}(\vx)[\sigma] := \begin{cases}
                           \vx_i'[\sigma]  & \qif I \succeq \bar\tau \text{ and } x_i[\tau] = 1 \\
                           \vx_i[\sigma] & \qq{otherwise}
                       \end{cases}
\end{align}

\begin{definition}
    Given a correlated profile $\mu$, a deviation $\phi$ of a player $i$ is {\em profitable} if the deviating player improves its expected utility: $\E_{\vx \sim \mu} u_i(\phi(\vx_i), \vx_{-i}) > \E_{\vx \sim \mu} u_i(\vx)$.
\end{definition}
\begin{definition}\label{def:equilibria mu}
    NFCCEs, EFCCEs, and EFCEs are correlated profiles $\mu$ that have no profitable deviations of their respective types.
\end{definition}

Here, in deciding whether to deviate, the players have common knowledge of the correlated profile $\mu$ from which their recommendations are drawn.

Given an objective function $g : \mc Z \to \R$, we say that an equilibrium  $\mu$ is {\em optimal} with respect to an objective $g : \mc Z \to \R$ if $\mu$ maximizes the expected objective value $\E_{\vx \sim \mu, z \sim \vx} g(z)$ among all equilibria of the same notion.

\begin{remark}
    The number of triggers available to a given player will play a fundamental role in the complexity of computing a solution according to each of the three solution concepts. In particular, for NFCCE, each player has only one trigger ($\Root_i$), whereas for EFCCE and EFCE, the number of triggers for each player depends on the depth of the game. We will see in \Cref{se:dag} that this difference results in a fundamental gap: under reasonable assumptions, an optimal NFCCE can be computed faster than an optimal EFCCE or an optimal EFCE.
\end{remark}

\subsection{Example of Solution Concepts}\label{se:example}
In this section, we give an example that illustrates the difference between NFCCE, EFCCE, and EFCE. Consider the extensive-form game in \Cref{fig:example-game}. As described in \Cref{example:efg}, this game represents a signalling game between two players, \pone and \ptwo. \ptwo has no rewards and will therefore never have incentives to deviate from recommendations. \pone scores a point if \ptwo plays the same action as chance played at the root, but chance's action is only privately revealed to \pone, so \ptwo relies on \pone to signal the chance action through \pone's own action. \pone also has the opportunity to receive a bonus point for guessing \ptwo's action in case \pfont2d is reached.

\newcommand{\pstrat}[5]{{\fbox{\text{\pfont1{#1#2}\pfont2{#3#4}\pfont1{#5}}}}}
We will refer to the pure profiles in this game using the notation
\pstrat bcdfh, where the letters indicate which actions were played at the respective infosets containing those nodes. For example, \pstrat LRLRL means that \pone plays left at \pfont1b, right at \pfont1c, and left at infoset \pfont1{hi}; while \ptwo plays left at \pfont2{de} and right at \pfont2{fg}---in particular, \pone copies chance, and \ptwo copies \pone. If \pone plays right at \pfont 1b, we leave \pone's action at \pfont 1{hi} unspecified since it is irrelevant; for example, \pstrat RLRL{} is a valid pure strategy.

We make the following observations about our example game.
\begin{itemize}
    \item The correlated profile $\mu_1 := \frac12 \pstrat LRLRR + \frac12 \pstrat RLRL{}$ is an NFCCE: \pone is getting utility $1$, which is larger than any utility it can get by unilaterally deviating without seeing any recommendations: since \ptwo's marginal strategy is uniform random, a best unilateral deviation for \pone is to always play left, securing expected utility $3/4$. However, $\mu_1$ is not an EFCCE, because \pone can profitably deviate at trigger \pfont1{hi} by playing left instead of right. This deviation cannot be expressed as an NFCCE deviation, because it requires \pone to follow recommendations at \pfont1{b} and \pfont1{c}.
    \item The correlated profile $\mu_2 := \frac12 \pstrat LRLRL + \frac12 \pstrat RRRR{}$ is an EFCCE. \pone still gets total expected utility 1. \pone is already getting the optimal utility at \pfont1c and \pfont1{hi}; and at \pfont1b, \pone is currently getting a conditional utility of $1$, and she cannot improve upon this without seeing the recommendation at \pfont1b. However, $\mu_2$ is not an EFCE, because \pone can profitably deviate upon being recommended to play right at \pfont1b by instead playing left at \pfont1b and right at \pfont1{hi}. This deviation cannot be expressed as an EFCCE deviation, because, in the deviation, \pone conditions her action at infoset \pfont1{hi} on the recommendation that she received at \pfont1b.
    \item The pure profile $\pstrat LRLRL$ is an EFCE (in fact, being uncorrelated, it is a Nash equilibrium).
\end{itemize}

\section{Unifying Correlated Solution Concepts via Mediator-Augmented Games}\label{se:cp}

As mentioned in the previous section and summarized in \Cref{fig:difference between correlated solution concepts}, different correlated solution concepts for extensive-form games differ in what the mediator (correlation device) reveals to the players, and whether the players' choices to commitment to follow the recommended behavior happen before or after observing the recommendation. These differences not only materialize in different equilibrium sets, but---as we will show later in this paper---also in complexity barriers that separate the solution concepts. Consequently, a unified treatment of these solution concepts needs to be approached with care.

In this section, we define augmented games in which the mediator is made explicit, which will be pivotal to our main results. Prior to presenting a precise formalization of the notion of augmented game, we provide some intuition about how the game is constructed, using the illustrative example in \Cref{fig:example-game}. During this phase, our primary objective is to provide a straightforward intuition about the construction process, deliberately omitting certain significant details that will be formally defined in \Cref{def:augmented game}. The augmented game explicitly represents players' choices regarding whether to adhere to mediators' recommendations or to deviate from them. Consequently, the augmented games will have different structures depending on which solution concept is desired---we will define one augmented game $\Gamma^c$ for each of our target solution concepts $c$. \Cref{fig:gamma c} summarizes the main connections between the computation of an optimal correlated concept $c$ (for instance, EFCE) in the original game $\Gamma$, and the computation of a Stackelberg equilibrium in the mediator-augmented game $\Gamma^c$ corresponding to $c$. \Cref{fig:augmented-example} depicts the augmented games derived from the example of \Cref{fig:example-game} for the three solution concepts of interest.

\begin{figure}[h!]
    \centering
    \begin{tikzpicture}
        \filldraw[draw=black,rounded corners=5pt,fill=black!5] (-7.26,-3.7) rectangle (-0.01, 0);
        \filldraw[draw=black,rounded corners=5pt,fill=black!5] (1.41,-3.7) rectangle +(7.25, 3.7);
        \node[draw,fill=black!12,semithick,text width=6.99cm,inner ysep=.10cm,align=center,anchor=east] (X) at (0,0) {Mediator-augmented game $\Gamma^c$};
        \node[draw,fill=black!12,semithick,text width=6.99cm,inner ysep=.10cm,align=center,anchor=west] (Y) at (1.4,0) {Original game $\Gamma$};
        \draw[implies-implies,double equal sign distance,thick] (0,-2) -- +(1.4,0);
        \node[anchor=north] at (.7,-.3) {
            \renewcommand\arraystretch{1.1}
            \setlength{\tabcolsep}{0.95cm}\small
            \begin{tabular}{@{}p{6.8cm}p{6.8cm}@{}}
                $\bullet$ Strategy for mediator $\vec\xi$ & $\bullet$ Correlation plan \\
                $\bullet$ Strategy for deviator player $\vec x^c_i$ & $\bullet$ Deviation strategy from recommendation \\
                $\bullet$ Obedient strategy for deviator $\vec{o}_i$ & $\bullet$ Identity deviation (\textit{i.e.}, no deviation) \\
                $\bullet$ Utility of mediator $\vec g$ & $\bullet$ Optimization objective (\textit{e.g.}, social welfare) \\ 
                $\bullet$ Utility of deviator $i$ & $\bullet$ Utility of $i$ in $\Gamma$ \\
                $\bullet$ $\vec\xi$ such that $(\vec o_i)_{i\in[n]}$ is a Nash strategy profile &  $\bullet$ Solution concept $c$ (\textit{e.g.}, EFCE) \\ 
                $\bullet$ Stackelberg equilibrium (solution to \eqref{eq:orig}) & $\bullet$ Optimal $c$ (\textit{e.g.}, welfare-maximizing EFCE)
            \end{tabular}
        };
    \end{tikzpicture}
    \caption{Correspondence between notions in the mediator-augmented game, and notions in the original game.}
    \label{fig:gamma c}
\end{figure}

In all three augmented games, the mediator has {\em imperfect recall}. This is crucial to correctly capture the correlated solution concepts. The imperfect recall is necessary for the one-to-one correspondence between {\em 
mixed strategies for the mediator} in the augmented game, and {\em correlated profiles of the players} in the original game. Intuitively, this is because the mediator's decisions in the augmented game correspond to {\em recommendations} in the original game, and therefore the mediator must pick one and only one recommendation in each information set. Thus, the mediator must have one infoset in the augmented game corresponding to each infoset in the original game. If the mediator were to have perfect recall, it would have the ability to ``break'' information sets by sending recommendations to a player that depend on information not known to that player. Therefore, there could be a strategy for the mediator that does not correspond to a strategy profile in the original game. 

\begin{description}
\item{\textbf{NFCCE.}} In the case of NFCCE (\Cref{fig:augmented-example}, top), the augmented game has an initial phase in which \pone and \ptwo decide whether to deviate or obey to the mediator. Only one player is allowed to deviate in the game. When player \pone (resp., \ptwo) deviates, all subsequent infosets will belong to either \pone (resp., \ptwo) or to the mediator. The mediator takes decisions on behalf of the obedient player. If both players are obedient (see the subtree with leaf nodes {\sf p,q,r,s,j,k,l,m,n,o}), then all decisions after the initial phase are taken by the mediator. %

\item{\textbf{EFCCE.}} In the case of EFCCE (\Cref{fig:augmented-example}, middle) we can reason as follows: starting from the root of the original game $\Gamma$, we replace each information set of player \pone or \ptwo with three new infosets. The first one is a parent infoset modelling the decision of the player to obey or to deviate at the original infoset in $\Gamma$. The two children information sets encode the decision to be taken at the original infoset of $\Gamma$ being replaced. The new infoset following from the decision of the player to obey (at the parent infoset) belongs to the mediator, who takes the action on behalf of the player. The new infoset following from the decision of the player to deviate (at the parent infoset) belongs to the deviating player, and it allows them for choosing the desired deviation. As before, after one player deviated, all the subsequent information sets belong to that player or to the mediator. 

\item{\textbf{EFCE.}} In the case of EFCE (\Cref{fig:augmented-example}, bottom), each of the original infosets $I$ of $\Gamma$ is duplicated and preceded by an information set of the mediator explixitly encoding the recommendation being issued at $I$. After observing the recommendation, the player decides whether to deviate or not.  We note that actions within the same mediator's information set can represent recommendations as well as actions taken on behalf of the player. This is contingent upon whether the other player previously made a deviation or not. The information available to the mediator when recommending actions or acting on behalf of a player remains identical to what the player would have had in the original game.
\end{description}

Following this intuition, given a game $\Gamma$, we define the \emph{augmented game} $\Gamma^c$ corresponding to solution concept $c$ as follows.

\begin{definition}\label{def:augmented game}
    Given an extensive-form game $\Gamma$, a solution concept $c$, and an objective function $g : \mc Z \to \R$, the {\em augmented game} $\Gamma^c$ is defined as follows.
    \begin{itemize}
        \item {\em Players.} $\Gamma^c$ has $n+1$ players: the $n$ players in $\Gamma$, and a {\em mediator}.
        \item {\em Histories.} Unless otherwise stated, histories in $\Gamma^c$ are identified with tuples $(h, a, \tau)$, where:
              \begin{itemize}[label={\scalebox{.8}{\textbullet}}]
                  \item $h$ is a history in $\Gamma$,
                  \item $a$ is either nothing $(\bot)$, a special symbol $*$, or an action $a \in A_h$, and
                  \item $\tau$ is either nothing $(\bot)$ or a trigger.
              \end{itemize}
        Intuitively, the three components of the history represent the following.
        \begin{itemize}[label={\scalebox{.8}{\textbullet}}]
            \item $h$ is the true history of the game, representing the actions that have been taken by the players.
            \item $a$ is the recommendation from the mediator at the current infoset. $\bot$ means that the mediator has yet to make a recommendation. $*$ means that the mediator has not given a recommendation. 
            \item $\tau$ represents the trigger, if any, that has been activated. Since we need only consider deviations of one player at a time, there can be at most one active trigger---$\bot$ means there is no active trigger.
        \end{itemize}
        \item {\em Terminal nodes.} If $z$ is terminal in $\Gamma$, then $(z, \bot, \tau)$ is terminal in $\Gamma^c$ for any $\tau$. At terminal node $z$, each player $i$ receives utility $u_i(z)$, and the mediator receives utility $g(z)$.
        \item {\em Chance nodes.} If $h$ is a chance node in $\Gamma$, then $(h, \bot, \tau)$ is also a chance node in $\Gamma^c$ with the same action probabilities. Chance selects an action $a \in A_h$, and the next node is $(ha, \bot, \tau)$. If $c = \nfcce$, the next node is $(ha, \bot, \tau)$. If $c \ne \nfcce$, the next node is a dummy node at which the sole action leads to $(ha, \bot, \tau)$.\footnote{The sole purpose of this ``dummy layer'' is to maintain \revision{timeability of $\Gamma^c$: for $c = \efcce$ and $c = \efce$, each layer of the game tree of $\Gamma$ needs to become {\em two} layers in $\Gamma^c$.}}
        \item {\em Information.} Players $i$ other than the mediator have perfect recall, and their observations are specified in the game description. The mediator does {\em not} have perfect recall: two histories $(h_1, \cdot, \cdot)$ and $(h_2, \cdot, \cdot)$ belong to the same mediator infoset in $\Gamma^c$ if and only if $h_1$ and $h_2$ belong to the same infoset in $\Gamma$.
\end{itemize}
\revision{
The remainder of the construction depends on the solution concept $c$.

For \nfcce:
    \begin{itemize}
    \item {\em Pre-play phase.} There is a pre-play phase in which each of the $n$ players, in order, chooses whether or not to deviate. Only one player may deviate: if player $i$ deviates, then players $j > i$ are forced to not deviate. Hence this pre-play phase is a tree with $n+1$ layers and $n+1$ leaves. The leaf in which player $i$ deviates is $(\Root, \bot, \Root_i)$, and the leaf in which no player deviates is $(\Root, \bot, \bot)$.
        \item {\em Player nodes $(h, \bot, \tau)$}  (with $h \in \mc H_i$ and $i \ne 0$). Suppose $h \in \mc H_i$ with $i \ne 0$. If $\tau = \Root_i$ then player $i$ observes the infoset $I \ni h$ and picks an action $a \in A_h$. Otherwise, the mediator acts by picking the action $a \in A_h$. In either case the next node is $(ha, \bot, \tau)$.
        \item There are no nodes $(h, a, \tau)$ with $a \ne \bot$.
    \end{itemize}
    
For \efcce:
    \begin{itemize}
        \item There is no pre-play phase.
        \item {\em Player nodes $(h, \bot, \tau)$} (with $h \in \mc H_i$ and $i \ne 0$). If $\tau \ne \bot$, then $(h, \bot, \tau)$ is a chance node with only one action, leading to node $(h, *, \tau)$. If $\tau = \bot$, then player $i$ observes the infoset $I \ni h$ decides whether or not to deviate. If player $i$ deviates, then the next node is $(h, *, I)$. Otherwise, the next node is $(h, *, \bot)$.
        \item {\em Player nodes $(h, a, \tau)$} (with $h \in \mc H_i$, $i \ne 0$, and $a \ne \bot$). If $\tau$ is a trigger of player $i$, then player $i$ selects an action $a \in A_h$. Otherwise, the mediator selects an action $a \in A_h$. In either case the next node is $(ha, \bot, \tau)$.
    \end{itemize}

    For EFCE:
    \begin{itemize}
        \item There is no pre-play phase.
        \item {\em Player nodes $(h, \bot, \tau)$} (with $h \in \mc H_i$ and $i \ne 0$). If $\tau$ is a trigger of player $i$, then $(h, \bot, \tau)$ is a chance node with only one action, leading to $(h, *, \tau)$. Otherwise, the mediator selects an action $a \in A_h$, and the next node is $(h, a, \tau)$.
        \item {\em Player nodes $(h, a, \tau)$} (with $h \in \mc H_i$, $i \ne 0$, and $a \ne \bot$). If $\tau \ne \bot$ is a trigger not belonging to player $i$, then $(h, a, \tau)$ is a chance node with a single action, leading to $(ha, \bot, \tau)$. Otherwise, player $i$ observes the infoset $I \ni h$ and action recommendation $a$, and selects an action $a' \in A_h$. If $a' = a$ or $a = *$, then the next node is \revision{$(ha', \bot, \tau)$}; otherwise, the next node is $(ha', \bot, Ia)$.
    \end{itemize}
}
\end{definition}
For concreteness,  in \Cref{fig:augmented-example} we show the augmented games derived from the example game in \Cref{fig:example-game} for all three solution concepts. For notational shorthand, we will use $h^\tau$ to refer to the node in $\Gamma^c$ corresponding to the mediator making a recommendation with history $h$ and trigger $\tau$---that is, $h^\tau = (h, \bot, \tau)$ when $c = \nfcce$ or $\efce$, and $h^\tau = (h, *, \tau)$ when \revision{$c = \efcce$}.

\begin{figure}
    \centering
    \tikzset{
        every path/.style={-},
        every node/.style={draw},
        infoset1/.style={-, densely dotted, ultra thick, color=p1color},
        infoset2/.style={infoset1, color=p2color},
        infosetm/.style={infoset1, color=p4color},
        terminal/.style={},
    }
    \forestset{
        default preamble={for tree={
        parent anchor=south, child anchor=north, s=0pt,
        s sep=0pt, l=24pt, l sep=0pt,
}},
        p1/.style={
            regular polygon,
            regular polygon sides=3,
            rounded corners=1.2,
            inner sep=2pt, fill=p1color, draw=none},
        p2/.style={p1, shape border rotate=180, fill=p2color},
        med/.style={circle,inner sep=3pt, fill=p4color, draw=none},
        parent/.style={no edge,tikz={\draw (#1.parent anchor) to (!.child anchor);}},
        parentd/.style n args={2}{no edge,tikz={\draw[ultra thick, p#2color] (#1.parent anchor) to (!.child anchor); }},
        dum/.style={inner sep=0pt, minimum size=0.1px, circle, fill=black},
        nat/.style={},
        terminal/.style={draw=none, font=\scriptsize\sf, inner sep=2pt},
        el/.style={edge label={node[midway, fill=white, inner sep=1pt, draw=none] {\scriptsize\sf #1}}},
        d/.style={edge={ultra thick, draw={p#1color}}},
        comment/.style={no edge, draw=none, align=center, font=\tiny\sf},
    }
    \begin{forest}
    [,draw=none
    [{\pone decides \\ whether to deviate},comment
    [{if \pone did not deviate, \\ \ptwo decides whether to deviate},comment
    [{the game is played. \\ the mediator acts for all \\ players who did not deviate},comment]]]
    [,p1,calign=first,no edge
        [,p2,calign=last,d=1
            [,nat
                [,med,name=fdl
                    [,p2,name=fdll
                        [,med,name=fdlll
                            [$\mathsf{p^2}$,terminal]
                            [$\mathsf{q^2}$,terminal]
                        ]
                        [,med,name=fdllr
                            [$\mathsf{r^2}$,terminal]
                            [$\mathsf{s^2}$,terminal]
                        ]
                    ]
                    [,p2,parent=fdr,name=fdrl
                        [$\mathsf{j^2}$,terminal]
                        [$\mathsf{k^2}$,terminal]
                    ]
                ]
                [,med,name=fdr
                    [,p2,parent=fdl,name=fdlr
                        [$\mathsf{l^2}$,terminal]
                        [$\mathsf{m^2}$,terminal]
                    ]
                    [,p2,name=fdrr
                        [$\mathsf{n^2}$,terminal]
                        [$\mathsf{o^2}$,terminal]
                    ]
                ]
            ]
            [,nat,d=2
                [,med,name=ffl
                    [,med,name=ffll
                        [,med,name=fflll
                            [$\mathsf{p}$,terminal]
                            [$\mathsf{q}$,terminal]
                        ]
                        [,med,name=ffllr
                            [$\mathsf{r}$,terminal]
                            [$\mathsf{s}$,terminal]
                        ]
                    ]
                    [,med,parent=ffr,name=ffrl
                        [$\mathsf{j}$,terminal]
                        [$\mathsf{k}$,terminal]
                    ]
                ]
                [,med,name=ffr
                    [,med,parent=ffl,name=fflr
                        [$\mathsf{l}$,terminal]
                        [$\mathsf{m}$,terminal]
                    ]
                    [,med,name=ffrr
                        [$\mathsf{n}$,terminal]
                        [$\mathsf{o}$,terminal]
                    ]
                ]
            ]
        ]
        [,dum
            [,nat
                [,p1,name=dfl
                    [,med,name=dfll
                        [,p1,name=dflll
                            [$\mathsf{p^1}$,terminal]
                            [$\mathsf{q^1}$,terminal]
                        ]
                        [,p1,name=dfllr
                            [$\mathsf{r^1}$,terminal]
                            [$\mathsf{s^1}$,terminal]
                        ]
                    ]
                    [,med,parent=dfr,name=dfrl
                        [$\mathsf{j^1}$,terminal]
                        [$\mathsf{k^1}$,terminal]
                    ]
                ]
                [,p1,name=dfr
                    [,med,parent=dfl,name=dflr
                        [$\mathsf{l^1}$,terminal]
                        [$\mathsf{m^1}$,terminal]
                    ]
                    [,med,name=dfrr
                        [$\mathsf{n^1}$,terminal]
                        [$\mathsf{o^1}$,terminal]
                    ]
                ]
            ]
        ]
    ]
    ]
    \draw[infosetm] (fflll) to (ffllr);
    \draw[infosetm] (ffll) to (ffrl);
    \draw[infosetm] (fflr) to (ffrr);
    \draw[infoset2] (fdll) to (fdrl);
    \draw[infoset2] (fdlr) to (fdrr);
    \draw[infosetm] (fdlll) to (fdllr);
    \draw[infosetm, bend right=30] (fdllr) to (fflll);
    \draw[infosetm, bend left=20] (ffl) to (fdl);
    \draw[infosetm, bend right=20] (ffr) to (fdr);
    \draw[infoset1] (dflll) to (dfllr);
    \draw[infosetm] (dfll) to (dfrl);
    \draw[infosetm] (dflr) to (dfrr);
    \draw[infosetm, bend left=30] (dfll) to (ffrl);
    \draw[infosetm, bend right=30] (dflr) to (ffrr);
    \end{forest}
    \begin{forest}
    [,draw=none
    [{nature acts \\ at the root},comment
    [{\pone decides \\ whether to deviate},comment
    [{if \pone has deviated, \pone acts; \\ else, the mediator acts for \pone},comment
    [{if nobody has deviated yet, \\ \ptwo decides whether to deviate},comment
    [{if \ptwo has deviated, \ptwo acts; \\ else, the mediator acts for \ptwo},comment
    [{if nobody has deviated yet, \\ \pone decides whether to deviate},comment
    [{if \pone has deviated, \pone acts; \\ else, the mediator acts for \pone},comment
    ]]]]]]]
    [,nat,no edge
        [,p1
            [,med,d=1,name=lm
                [,p2,name=ll
                   [,p2,name=lld
                        [,dum,name=lldl
                            [,med,name=lldlm
                                [$\mathsf{p^{d}}$,terminal]
                                [$\mathsf{q^{d}}$,terminal]
                            ]
                        ]
                        [,dum,name=lldr
                            [,med,name=lldrm
                                [$\mathsf{r^{d}}$,terminal]
                                [$\mathsf{s^{d}}$,terminal]
                            ]
                        ]
                    ]
                    [,p2,name=rld,parent=rl
                        [$\mathsf{j^{e}}$,terminal]
                        [$\mathsf{k^{e}}$,terminal]
                    ]
                ]
                [,p2,parent=rm,name=rl
                    [,med,parentd={ll}{2},name=llm
                        [,p1,name=lll
                            [,med,d=1,name=lllm
                                [$\mathsf{p}$,terminal]
                                [$\mathsf{q}$,terminal]
                            ]
                            [,med,name=llrm,parentd={llr}{1}
                                [$\mathsf{r}$,terminal]
                                [$\mathsf{s}$,terminal]
                            ]
                        ]
                        [,p1,name=llr
                            [,p1,parent=lll,name=llld
                                [$\mathsf{p^h}$,terminal]
                                [$\mathsf{q^h}$,terminal]
                            ]
                            [,p1,name=llrd
                                [$\mathsf{r^i}$,terminal]
                                [$\mathsf{s^i}$,terminal]
                            ]
                        ]
                    ]
                    [,med,d=2,name=rlm
                        [$\mathsf{j}$,terminal]
                        [$\mathsf{k}$,terminal]
                    ]   
                ]
            ]   
            [,p1,name=ld
                [,dum,name=ldl
                    [,med,name=ldlm
                        [,dum,name=ldll
                            [,p1,name=ldlld
                                [$\mathsf{p^b}$,terminal]
                                [$\mathsf{q^b}$,terminal]
                            ]
                        ]
                        [,dum,name=ldlr
                            [,p1,name=ldlrd
                                [$\mathsf{r^b}$,terminal]
                                [$\mathsf{s^b}$,terminal]
                            ]
                        ]
                    ]
                ]
                [,dum,parent=rd,name=rdl
                    [,med,name=rdlm
                        [$\mathsf{j^c}$,terminal]
                        [$\mathsf{k^c}$,terminal]
                    ]   
                ]
            ]
        ]
        [,p1
            [,med,d=1,name=rm
                [,p2,name=lr,parent=lm
                   [,p2,name=lrd
                        [$\mathsf{l^f}$,terminal]
                        [$\mathsf{m^f}$,terminal]
                    ]
                    [,p2,name=rrd,parent=rr
                        [$\mathsf{n^g}$,terminal]
                        [$\mathsf{o^g}$,terminal]
                    ]
                ]
                [,p2,name=rr
                    [,med,parentd={lr}{2},name=lrm
                        [$\mathsf{l}$,terminal]
                        [$\mathsf{m}$,terminal]
                    ]
                    [,med,d=2,name=rrm
                        [$\mathsf{n}$,terminal]
                        [$\mathsf{o}$,terminal]
                    ]   
                ]
            ]
            [,p1,name=rd
                [,dum,name=ldr,parent=ld
                    [,med,name=ldrm
                        [$\mathsf{l^b}$,terminal]
                        [$\mathsf{m^b}$,terminal]
                    ]
                ]
                [,dum,name=rdr
                    [,med,name=rdrm
                        [$\mathsf{n^c}$,terminal]
                        [$\mathsf{o^c}$,terminal]
                    ]   
                ]
            ]
        ]
    ]
    ]
    \draw[infosetm] (lldlm) to (llrm);
    \draw[infoset1] (llrd) to (llld);
    \draw[infoset1] (lll) to (llr);
    \draw[infoset2] (ll) to (rl);
    \draw[infoset2] (lld) to (rld);
    \draw[infosetm] (llm) to (rdlm);
    \draw[infoset1] (ldlld) to (ldlrd);
    \draw[infoset2] (lr) to (rr);
    \draw[infoset2] (lrd) to (rrd);
    \draw[infosetm] (lrm) to (rdrm);
    \end{forest}
    \begin{forest}
    [,draw=none
    [{nature acts \\ at the root},comment
    [{mediator picks \\ a rec for \pone},comment
    [{\pone observes rec \\ and takes action},comment
    [{mediator picks \\ a rec for \ptwo},comment
    [{\ptwo observes rec \\ and takes action \\ (can only deviate if \\ \pone has not deviated)},comment
    [{if \pone has not deviated, \\ mediator picks
    \\ a rec for \pone},comment
    [{\pone observes rec \\ and takes action \\ (can only deviate if \\ \ptwo has not deviated)},comment
    ]]]]]]]
    [,nat,no edge
        [,med,name=l
            [,p1,name=lml
                [,med,d=1,name=ll
                    [,p2,name=llml
                        [,med,d=2,name=lll
                            [,p1,name=lllml
                                [$\mathsf{p}$,d=1,terminal]
                                [$\mathsf{q^{h\rL}}$,terminal]
                            ]
                            [,p1,name=lllmr
                                [$\mathsf{p^{h\rR}}$,terminal]
                                [$\mathsf{q}$,d=1,terminal]
                            ]
                        ]
                        [,med,name=lldr
                            [,dum [$\mathsf{r^{d\rL}}$,terminal]]
                            [,dum [$\mathsf{s^{d\rL}}$,terminal]]
                        ]
                    ]
                    [,p2,name=llmr
                        [,med,name=lldl
                            [,dum [$\mathsf{p^{d\rR}}$,terminal]]
                            [,dum [$\mathsf{q^{d\rR}}$,terminal]]
                        ]
                        [,med,d=2,name=llr
                            [,p1,name=llrml
                                [$\mathsf{r}$,d=1,terminal]
                                [$\mathsf{s^{i\rL}}$,terminal]
                            ]
                            [,p1,name=llrmr
                                [$\mathsf{r^{i\rR}}$,terminal]
                                [$\mathsf{s}$,d=1,terminal]
                            ]
                        ]
                    ]
                ]
                [,med,name=rdl,parent=lmr
                    [,dum
                        [,dum
                            [,p1,name=rdll
                                [$\mathsf{p^{b\rR}}$,terminal]
                                [$\mathsf{q^{b\rR}}$,terminal]
                            ]
                        ]
                    ]
                    [,dum
                        [,dum
                            [,p1,name=rdlr
                                [$\mathsf{r^{b\rR}}$,terminal]
                                [$\mathsf{s^{b\rR}}$,terminal]
                            ]
                        ]
                    ]
                ]
            ]
            [,p1,name=rlml,parent=r
                [,med,name=rll,parentd={rlml}{1}
                    [,p2,name=rllml
                        [$\mathsf{j}$,d=2,terminal]
                        [$\mathsf{k^{e\rL}}$,terminal]
                    ]
                    [,p2,name=rllmr
                        [$\mathsf{j^{e\rR}}$,terminal]
                        [$\mathsf{k}$,d=2,terminal]
                    ]
                ]
                [,med,name=rrdl,parent=rlmr
                    [,dum[$\mathsf{j^{c\rR}}$,terminal]]
                    [,dum[$\mathsf{k^{c\rR}}$,terminal]]
                ]
            ]
        ]
        [,med,name=r
            [,p1,name=lmr,parent=l
                [,med,name=ldr,parent=lml
                    [,dum [$\mathsf{l^{b\rL}}$, terminal]]
                    [,dum [$\mathsf{m^{b\rL}}$, terminal]]
                ]
                [,med,name=rr,parentd={lmr}{1}
                    [,p2,name=rrml
                        [$\mathsf{l}$,d=2, terminal]
                        [$\mathsf{m^{f\rL}}$, terminal]
                    ]
                    [,p2,name=rrmr
                        [$\mathsf{l^{f\rR}}$, terminal]
                        [$\mathsf{m}$,d=2, terminal]
                    ]
                ]
            ]
            [,p1,name=rlmr
                [,med,name=rldr,parent=rlml
                    [,dum [$\mathsf{n^{c\rL}}$, terminal]]
                    [,dum [$\mathsf{o^{c\rL}}$, terminal]]
                ]
                [,med,name=rrr,d=1
                    [,p2,name=rrrml
                        [$\mathsf{n}$,d=2, terminal]
                        [$\mathsf{o^{g\rL}}$, terminal]
                    ]
                    [,p2,name=rrrmr
                        [$\mathsf{n^{g\rR}}$, terminal]
                        [$\mathsf{o}$,d=2, terminal]
                    ]
                ]
            ]
        ]
        ]
    ]
    ]
    \draw[infosetm] (lll) to (llr);
    \draw[infosetm] (ll) to (rrdl);
    \draw[infoset1,bend right=15] (lllml) to (llrml);
    \draw[infoset1,bend left=20] (lllmr) to (llrmr);
    \draw[infoset1] (rdll) to (rdlr);
    \draw[infosetm] (ldr) to (rrr);
    \draw[infoset2, bend right=15](llml) to (rllml);
    \draw[infoset2, bend left=15](llmr) to (rllmr);
    \draw[infoset2, bend right=30](rrml) to (rrrml);
    \draw[infoset2, bend left=30](rrmr) to (rrrmr);
    \end{forest}
        \caption{Augmented games $\Gamma^c$ for NFCCE (top), EFCCE (center), and EFCE (bottom), where $\Gamma$ is the example game in \Cref{fig:example-game}. The obedient strategies $\vec o_1, \vec o_2$ are given by the thick colored lines below \pone and \ptwo's decision points. Red circles denote decision points of the mediator. Augmented histories are labeled as $h^\tau$, where $h$ is the true node and $\tau$ is the trigger. If no superscript is present, there was no trigger. For cleanliness, $\tau$ is abbreviated in all three diagrams. For NFCCE, $\tau$ is the player $i$ who deviated---for example, $\mathsf{p^2}$ means terminal node {\sf p} was reached, but \ptwo deviated. For EFCCE, $\tau$ is the node at which the player deviated---for example, $\mathsf{p^d}$ means terminal node {\sf p} was reached, but $\ptwo$ deviated at node \pfont2{d}. For EFCE, $\tau$ is the node at which the player deviated, followed by the recommendation ($\rL$ or $\rR$) given to the player at that node---for example, $\mathsf{q^{h\rL}}$ means terminal node {\sf q} was reached but \pone deviated after being recommended to play $\rL$ at \pfont1{h}.}
        \label{fig:augmented-example}
\end{figure}
\subsection{Optimal Correlation via the Augmented Game}\label{sec:lp}
We now discuss how to use the augmented game $\Gamma^c$ to compute optimal correlated equilibria in $\Gamma$. We first make a few critical observations:

First, the mediator has exactly one information set corresponding to each information set of the original game $\Gamma$. Therefore, {\em pure strategies} of the mediator correspond to {\em pure profiles} in $\Gamma$, and {\em mixed strategies} of the mediator correspond to {\em correlated profiles} in $\Gamma$. We will therefore abuse notation and also use $\vec\xi$ to refer to mixed strategies for the mediator in $\Gamma$. Critically, the sequence form of $\xi$ in each augmented game will have enough information about the correlated distribution to define the incentive constraints of the players.
Second, each player has a unique {\em obedient strategy} $\vec o_i$, defined by always obeying recommendations (for EFCE) and never choosing to deviate (or NFCCE and EFCCE). 
Finally, the size of the $\Gamma^c$ is polynomial in the size of $\Gamma$.

As a notational convention, where context is insufficient, we will generally use a superscript $c$ to distinguish the augmented game from the original game---for example, $\X_i^c$ will denote the strategy set of player $i$ in $\Gamma^c$, {\em etc.}

Now let $\vec\xi$ be a mediator mixed strategy in $\Gamma^c$. Then $\vec\xi$ represents an equilibrium in $\Gamma$ if and only if, in the profile $(\vec\xi, \vec o_1^c, \dots, \vec o_n^c)$, each (non-mediator) player $i$ is playing a best response. That is, solving the following program will give an optimal equilibrium:
\begin{align}
    \max_{\vec\xi \in \Xi^c} \quad g(\vec\xi) \qq{s.t.} \max_{\vec x_i^c \in \mc X_i^c} u_i(\vec\xi, \vec x_i^c, \vec o_{-i}^c) \le u_i(\vec\xi, \vec o_i^c, \vec o_{-i}^c) \quad \forall i \in [n]
\end{align}
where $\Xi^c$ is the mediator's sequence-form mixed strategy set in $\Gamma^c$, and $\mc X_i^c$ is player $i$'s mixed strategy set in $\Gamma^c$. Now, by representing the mixed strategy of each player (including the  mediator) in sequence form, the utility functions are linear in each strategy. Therefore, the above program can be rewritten as
\begin{align}
    \max_{\vec\xi \in \Xi^c} \quad \vec g^\top\vec\xi \qq{s.t.} \max_{\vec x_i^c \in \mc X_i^c} \vec \xi^\top \vec A_i \vec x_i^c \le \vec b_i^\top \vec\xi  \quad \forall i \in [n]\label{eq:orig}
\end{align}
for vectors and matrices $\vec g, \vec A_i$, and $\vec b_i$. Now, the inner maximization
\begin{align}
    \max_{\vec x_i^c \in \mc X_i^c} \vec \xi^\top \vec A_i \vec x_i^c \label{eq:inner}
\end{align}
is itself an LP where $\vec\xi$ is a constant. Moreover, since each player $i$ has perfect recall, the sequence-form strategy sets $\mc X_i^c$ can be represented as polytopes $\mc X_i^c = \{ \vec x_i^c \ge \vec 0: \vec F_i^c \vec x_i^c = \vec f_i^c \}$ for matrix and vector $\vec F_i^c, \vec f_i^c$ of size linear in the size of $\Gamma^c$. We therefore can formally take a dual of \eqref{eq:inner}, resulting in the LP
\begin{align}
    \min_{\vec v_i} (\vec f_i^c)^\top \vec v_i \qq{s.t.} \vec A_i^\top \vec \xi \le (\vec F_i^c)^\top \vec v_i. \label{eq:inner-dual}
\end{align}
By strong duality of linear programs (which holds in this case because \eqref{eq:inner} is always feasible), the programs \eqref{eq:inner} and \eqref{eq:inner-dual} have the same value. Therefore, \eqref{eq:orig} is equivalent to the linear program
\begin{align}\label{eq:lp}
\left\{\begin{aligned}
    \max_{\substack{\vec\xi, \vec v_i : i \in [n]}} & \quad \vec g^\top\vec\xi \\ \text{s.t.} \quad &\circled{1}~ \vec A^\top_i \vec\xi \le (\vec F_i^c)^\top \vec v_i&& \forall i \in [n] \\
            &\circled{2}~ (\vec f_i^c)^\top \vec v_i \le \vec b_i^\top \vec\xi&& \forall i \in [n]\\
    & \circled{$\star$}~~ \vec\xi \in \Xi^c
\end{aligned}\right. \tag*{LP}
\end{align}
This program has size linear in the size of $\Gamma^c$ and the description of the polytope $\Xi^c$. Unfortunately, in general, since the mediator has imperfect recall, there is no efficient way of representing $\Xi^c$, that is, there is no polynomial system of linear constraints describing $\Xi^c$. Indeed computing optimal equilibria for all three notions $c$ is NP-hard~\cite{Stengel08:Extensive}. 

Although the {\em pure strategy sets} for the mediator are essentially the same in all three augmented games, the {\em sequence-form strategy sets} $\Xi^c$ are substantially different. The differences arise due to more deviations being possible for some notions than for others. Consider for example the game $\Gamma$ depicted in \Cref{fig:example-game}. In the augmented game $\Gamma^\efce$ (\Cref{fig:augmented-example}, bottom),  there is a terminal node $\mathsf{p^{b\rR}}$ whose player reach probability $\xi[\mathsf{p^{b\rR}}]$ is the probability that the mediator recommends \pone to play right at \pfont1{b} and \ptwo to play left at infoset \pfont2{de}. There is no node in $\Gamma^\nfcce$ whose player reach probability represents the same thing. It should therefore remain intuitively plausible that $\Xi^\efce$ should be more difficult to represent than $\Xi^\nfcce$. In the next section, we will discover that this is precisely the case. 

\subsection{Comparison to Relevant Sequence-Based Construction of $\Xi$}\label{se:relevant}
Our construction via the mediator-augmented game uses a vector $\vec\xi \in \Xi^c$ to represent a correlated profile. It is instructive to compare this representation to other representations of correlated profiles, in particular, the {\em correlation plan} defined and used by \citet{Stengel08:Extensive}. In this section, we will review the notion of correlation plan defined by that paper, and compare it to our construction.

\begin{definition}\label{def:connected-infosets}
    A sequence tuple $(I_1 a_1, \dots, I_n a_n) \in \Sigma_1 \times \dots \times \Sigma_n$ is {\em relevant} if there is a history $h$ in $\Gamma$ such that either $\sigma_i(h) = I_i a_i$ for every player $i$, or there is a player $j$---the {\em deviator}---such that $\sigma_i(h) = I_i a_i$ for all $i \ne j$ and $I_j \preceq h$. 
\end{definition}

This definition was first proposed by \citet{Stengel08:Extensive} in the two-player case; here, we generalize it to arbitrarily many players. Intuitively, the relevant tuples are those that appear in the linear program defining {\em any} of the three notions.
\begin{definition}[\citealp{Stengel08:Extensive}]
    For a correlated profile $\mu \in \Delta(X_1 \times \dots \times X_n)$, the {\em correlation plan} is the vector $\vec\xi \in \R^{\Sigma}$ defined by $\vec\xi[\sigma_1, \dots, \sigma_n] = \E_{\vx \sim \mu} \prod_{i \in [n]} \vec x_i[\sigma_i]$. We denote by $\Xi$ the set of all correlation plans.
\end{definition}

\citet{Stengel08:Extensive} go on to show that correlation plans are a sufficient representation for computing (optimal) EFCE, in the sense that, if one could efficiently represent the set of all correlation plans, then one can compute optimal EFCE efficiently. Farina et al.~\cite{Farina19:Correlation,Farina20:Coarse} generalizes this observation to NFCCE and EFCCE as well. Our linear program \eqref{eq:lp} achieves the same claim: if $\Xi^c$ is efficiently representable then optimal equilibria in notion $c$ can be computed efficiently. One may wonder, therefore, about the relationship between the two.

It turns out that each of our $\Xi^c$ polytopes is in some sense merely a sub-vector of $\Xi$ with the indices renamed. That is, there is a natural injection from sequences of the mediator in $\Gamma^c$ to relevant tuples $(\sigma_1, \dots, \sigma_n) \in \Sigma$. %
A mediator sequence in $\Gamma^c$ corresponds to some history $h^\tau$. If $\tau = \bot$ then $h^\tau$ corresponds to $(\sigma_1(h), \dots,  \sigma_n(h))$, that is, $\vec\xi[h^\tau] = \vec\xi[\sigma_1(h), \dots,  \sigma_n(h)]$; if $\tau$ is a nonempty trigger (say, P1 WLOG), then $h^\tau$ corresponds to $(\sigma_1(\tau), \sigma_{2}(h), \dots, \sigma_n(h))$, where $\sigma_1(\tau)$ is the last sequence of player $i$ before $\tau$. By construction of $\Gamma^c$, this must be a relevant tuple.

In some sense, $\Xi^c$ is therefore a {\em refined} notion of correlation plan that is specific to the equilibrium concept $c$, only requiring the sequence tuples that are relevant for that concept. In the next section, we will show that, in fact, the differences between the various $\Xi^c$s result in separations in the complexity of representing each polytope, and therefore separations in the complexity of computing optimal equilibria.

The key barrier to computing optimal equilibria, in a sense, is that {\em the mediator in the augmented game has imperfect recall}. In the next two sections, we will describe two methods of overcoming this imperfect recall and thus of arriving at algorithms for computing optimal equilibria. The first (\Cref{se:dag}) applies the recent construction of \citet{Zhang22:Team_DAG}, which is a general method of representing the sequence form of an imperfect-recall player in a timeable game. The second (\Cref{se:colgen}) is a variant of {\em column generation} which is most powerful in two-player games, in which one (and only one) player is allowed to play a {\em mixed} strategy, thereby allowing a much greater strategy set to be available for any given support.

\section{Representing Imperfect-Recall Decision Spaces}\label{se:dag}
\citet{Zhang22:Team_DAG} recently developed a method for representing the sequence-form strategy spaces for imperfect-recall players (equivalently, teams of players who cannot communicate) in timeable games. Since the augmented games $\Gamma^c$ are timeable, we directly apply their main result to our problem.
\begin{definition}\label{def:connectivity}
    In a timeable extensive-form game $\Gamma'$, the {\em connectivity graph} $G_S$ of a subset of players $S \subseteq [n]$ is the graph whose nodes are histories of $\Gamma'$, and where there is an edge $(h, h')$ if $h$ and $h'$ are in the same level of the tree, {\em and} they are {\em connected}, {\em i.e.}, there is an infoset $I \in \mc I_i$, where $i \in S$, with $h, h' \preceq I$. \revision{We will use $G_i$ as shorthand for $G_{\{i\}}$}.
\end{definition}

\begin{definition}
    A set of nodes $B \subseteq \mc H$ is a {\em belief} for player $i$ if
    \begin{enumerate}
        \item $B$ contains at least one decision point for player $i$, that is, $B \cap \mc H_i \ne \emptyset$
        \item there exists a pure strategy\footnote{\revision{Recall that $\Pi_i$ is the set of pure strategies for player $i$, obeying any imperfect recall constraints.}} \revision{$\vec x_i \in \Pi_i$} for player $i$ such that $B$ is a connected component of $G_i[\{ h \in \mc H : \vx_i[h] = 1\}]$ where $G_i[\cdot]$ denotes an induced connected component of $G_i$.
    \end{enumerate}
\end{definition}
We will use $\mc B_i$ to denote the set of beliefs of player $i$.

Intuitively, beliefs represent sets of nodes that an imperfect-recall player {\em will always be able to distinguish in the future}: that is, if $B$ is a belief corresponding to pure strategy $\vec x_i$, then, upon reaching the belief $B$, player $i$ knows that it has reached belief $B$, and player $i$ knows that it will never forget having reached $B$. 
\begin{theorem}[Team Belief DAG \cite{Zhang22:Team_DAG}]\label{th:dag}
    There exists a representation of player $i$'s decision space as a polytope whose constraint matrix has $O^*(R_i)$ entries, where
    \begin{align}\label{eq:lp-size}
     R_i := \sum_{B \in \mc B_i} \prod_{\substack{I \in \mc I_i:\\ I \cap B \ne \emptyset}} |A_I|
    \end{align}
\end{theorem}
The representation uses a DAG to model the decision problem faced by player $i$, and then bounds the number of nodes in the DAG.
For intuition, when $\Gamma'$ has perfect recall, one can check that beliefs are always disjoint and every infoset $I \in \mc I_i$ is a belief, so the above expression is linear in the size of the game---indeed, in that case, the representation reduces to the sequence-form polytope. In the interest of self-containment, \Cref{app:tb-dag} contains a description of the construction of \citet{Zhang22:Team_DAG}.

We use the above result to construct a representation of the mediator's decision space, $\Xi^c$, in the augmented game $\Gamma' := \Gamma^c$. We call the representation of $\Xi^c$ using \Cref{th:dag} the {\em correlation DAG} for notion $c$. \Cref{th:dag} immediately gives an algorithm for solving the program \eqref{eq:lp}. This algorithm is given in \Cref{al:dag}. 
\begin{algorithm}[!ht]
\caption{Optimal Correlated Equilibria via Correlation DAG}
\label{al:dag}
{\bf input:} extensive-form game $\Gamma$, desired solution concept $c$, objective $g: \mc Z \to \R$\;
construct the augmented game $\Gamma^c$\;
compute a polytope representation of the mediator's strategy space, $\Xi^c$, using \Cref{th:dag}\;
solve the LP \eqref{eq:lp}\;
{\bf return} $\vec\xi$\;
\end{algorithm}

\subsection{Analyzing the Size of the Representation}

To analyze the complexity of \Cref{al:dag}, it suffices to bound the quantity in \eqref{eq:lp-size}. Notationally, we will use $R_\mediator^c$ to denote the quantity $R_\mediator$ in \eqref{eq:lp-size} in the augmented game $\Gamma^c$.  We first introduce some useful definitions.
\begin{definition}
    A {\em public state} is a connected component of $G = G_{[n]}$. \revision{We will use $\mc P$ to denote the set of public states.}
\end{definition}
\begin{definition}
    Given a node $h$ and a player $i$, the {\em last infoset} $I_i(h)$ is the lowest (\textit{i.e.}, most recent) infoset reached by player $i$ on the path to $h$.
\end{definition}
\begin{definition}
    The {\em information complexity} $k$ of an extensive-form game is the greatest number of unique last infosets in any public state. In symbols,
    $k = \max_{P \in \mc P} \abs{ \{ I_i(h) : h \in P, i \in [n] \} }$. 
\end{definition}

Notice that it is possible for $k$ to be much smaller than $n|P|$, because the set of last infosets may contain duplicates. For example, in normal-form games (converted to extensive form in the canonical manner), we have $k = n$ since each terminal node is a public state and each player has only one infoset. As an example, the information complexity of the game in \Cref{fig:example-game} is $3$: the public state \pfont2{de} has three last infosets, namely \pfont1{b}, \pfont1{c}, and \pfont2{de} itself.

\citet{Zhang22:Team_DAG} use the definition of information complexity to bound the representation size of \Cref{th:dag}. In particular, they show that if the decision problem for the imperfect-recall player $i$ can be decomposed into $n$ perfect-recall players such that the information complexity is $k$, then $R_i \le O^*((b+1)^k)$, where $b$ is the branching factor of the game. In this section, we show similar bounds in our setting. Note that $b$ and $k$ here are the branching factor and information complexity of the original game $\Gamma$, not of $\Gamma^c$---therefore, we cannot directly apply the bound $R_i \le O^*((b+1)^k)$. Indeed, the mediator in $\Gamma^c$ can have much higher information complexity than $\Gamma$. Thus, we need to be more careful in our analysis.

\begin{restatable}{theorem}{thmCorrelationDagSize}\label{th:dag size} Let $k$ be the information complexity of a timeable game $\Gamma$, $b$ be its branching factor, and $d$ be its depth. Then $R_{\mediator}^\nfcce  \le O^*\qty((b+1)^k)$, $R_{\mediator}^\efcce \le O^*\qty((b+d-1)^k)$, and  $R_{\mediator}^\efce \le O^*\qty((bd)^k)$. 
\end{restatable}
\begin{proof}
The expression \eqref{eq:lp-size} counts the number of pairs $(B, \vec a)$ where $B$ is a belief of the mediator $\Gamma^c$ and $\vec a \in \prod_{I \in \mc I_\mediator, I \cap B \ne\emptyset} A_I$. Our goal will therefore be to bound this number.

\vspace{.3cm}
\begin{itemize}[left=2cm]
    \item[NFCCE:] It suffices, for each last-infoset $I$ of $P$, to specify whether the player (1) does not play to $I$ at all, or (2) plays to $I$ and chooses one of the $b$ actions available therein. There are at most $(b+1)^k$ such choices. Each choice induces a disjoint collection of pairs $(B, \vec a)$; that is, surely at most $\abs{P}$ pairs. Thus, $R_\mediator^\nfcce = O^*((b+1)^k)$.
    \item[EFCCE:] For each of the $k$ last-infosets $I$ at $P$, we need to specify whether the player played to reach $I$ and then played one of the (at most) $b$ actions available there, or she deviated at one of the (at most) $d-2$ infosets $I' \prec I$. There are at most $b+d-1$ ways to do this, so, by the above argument, we have $R_\mediator^\efcce = O^*((b+d-1)^k)$.
    \item[EFCE:] For EFCE, we need to additionally specify which action was recommended at the deviation point, of which there are at most $b$ possibilities, for a total of $b(d-1)+b = bd$ options. Thus, again by the same argument, $R_\mediator^\efce = O^*((bd)^k)$. \qedhere
\end{itemize}
\end{proof}

As an example, consider an extensive-form game of the following form. Chance first samples and privately reveals types $t_i \in [T]$ to each player $i$. Thereafter, there is no further privacy: all actions by the players and chance after the root are public. By definition, we see that this game is a public-action game, and we have $k = nT$ because each sequence of post-root actions induces a public state with $T$ private states for each of the $n$ players. Thus, \Cref{th:dag size} gives an algorithm for computing optimal EFCEs that runs in time $\poly(\abs{\mc H}, (bd)^{nT})$; in particular, if $n = T = O(1)$ then the algorithm runs in polynomial time. To our knowledge, we are the first to give a polynomial-time algorithm for this setting, even when $n = T = 2$.

We now show two settings in which we can improve our bounds from \Cref{th:dag size}. They both depend on certain information being {\em public}.

\subsection{Public Player Actions}\label{se:public actions}
First, we discuss the setting in which {\em player} actions are public.%
\begin{definition}
A game has {\em public player actions} if, for all public states $P \in \mc P$ containing at least one non-chance node, for all actions $a \in \bigcup_{h \in P} A_h$, the set $\{ ha : h \in P, a \in A_h \}$ is a union of public states.
\end{definition}

Poker, for example, has this structure: the root public state contains only a chance node, and every action thereafter is fully public. In this setting, we can remove the dependencies on $b$ for NFCCE and EFCCE:

\begin{restatable}{theorem}{thmPublicActions}\label{th:public actions}
In games with public player actions, $R_\mediator^\nfcce = O^*\qty(3^k)$ and $R_\mediator^\efcce = O^*\qty(d^k)$.
\end{restatable}
Intuitively, the proof works by constructing a new game that reduces the branching factor of the original game to $2$ while keeping all other relevant structure intact. The fact that the players' actions are public ensures that this transformation does not increase $k$. We defer the full proof to \Cref{se:pr:public actions}, since it is similar to the proof in the team setting given by \citet{Zhang22:Team_DAG}. 

Once again, the bound for NFCCE matches that of \citet{Zhang22:Team_DAG} in team games, up to polynomial factors. The bound on $R^\efce_\mediator$ cannot be improved in this fashion, for two reasons. First, the $(bd)^k$ term in that analysis comes from counting the number of triggers at a given node, which has not changed. Second, as above, the proof of \Cref{th:public actions} modifies the original game tree to have lower branching factor. This is an invalid transformation for EFCE, because some EFCE triggers present in the original game would not be expressible in the new game.

\subsection{Two-Player Games with Public Chance}\label{se:pubchance}
We now discuss the case where {\em chance} actions are public. Since it is already NP-hard to compute optimal equilibria in three-player games with no chance nodes~\cite{Stengel08:Extensive}, we restrict our attention to {\em two-player} games. \citet{Farina20:Polynomial} showed, via a different construction, that in games with public chance, $\Xi$ has a polynomial-sized representation and therefore optimal NFCCEs, EFCCEs, and EFCEs can be computed in polynomial time. In this section, we show that our correlation DAG matches this bound.

\begin{definition}
    A game has {\em public chance actions} if, for every two nodes $h, h'$ in the same public state, the lowest common ancestor $h \land h'$ is not a chance node.
\end{definition}

We will assume for the rest of this section that levels in $\Gamma$ uniquely specify whose move it is---that is, for every level of the game tree, there exists a player $i$ (possibly nature) such that every node in the level is a decision node of player $i$.  Since we have already assumed timeability, this additional assumption is without loss of generality by adding dummy nodes~\cite{Carminati22:Public}. Most practical games, including the games we use in our experiments, already satisfy this assumption without further modification.

\begin{theorem}\label{th:pubchance}
    In two-player timeable games with public chance actions, we have $R_\mediator^c = \poly(\abs{\mc H})$ for all three notions $c$.
\end{theorem}

Initially, one may ask whether it is possible to prove this result by directly applying \Cref{th:dag size}. In particular, if it were the case that all two-player games of public chance had constant information complexity, \Cref{th:pubchance} would follow immediately. Unfortunately, this is not the case: in \Cref{fig:pubchance-high-k}, we exhibit two families of two-player extensive-form games with {\em no} chance actions and information complexity that is linear in the size of the game.

\begin{figure}[t]
\centering
\tikzset{
    every path/.style={-},
    infoset1/.style={-, densely dotted, ultra thick, color=p1color},
    infoset2/.style={-, densely dotted, ultra thick, color=p2color},
}
\forestset{
    default preamble={for tree={parent anchor=south, child anchor=north, s sep=2pt}},
    p1/.style={regular
        polygon, regular polygon
        sides=3, inner sep=2pt, fill=p1color, draw=none},
    p2/.style={p1, shape border rotate=180, fill=p2color},
    parent/.style={no edge,tikz={\draw (#1.parent anchor) to (!.child anchor);}},
}
\begin{forest}
[,p1
   [,p1 [,p2,name=1 [] []] [,p2 [] []]]
   [,p1 [,p2 [] []] [,p2 [] []]]
   [,p1 [,p2 [] []] [,p2,name=2 [] []]]
]
\draw[infoset2] (1) to (2);
\end{forest}
\qquad
\begin{forest}
[,p1
   [,p2
        [,p2,name=b1
            [,p1,name=a1 [] []]
            []
        ]
        [,p2
            [,p1,name=a2 [] []]
            []
        ]
   ]
   [,p2
        [,p2,name=b3
            [,p1,name=a3 [] []]
            []
        ]
        [,p2
            [,p1,name=a4 [] []]
            []
        ]
   ]
   [,p2
        [,p2,name=b5
            [,p1,name=a5 [] []]
            []
        ]
        [,p2
            [,p1,name=a6 [] []]
            []
        ]
   ]
]
\draw[infoset1, bend right=45] (a1) to (a2);
\draw[infoset1, bend right=45] (a3) to (a4);
\draw[infoset1, bend right=45] (a5) to (a6);
\draw[infoset2, bend right=45] (b1.center) to (b3.center) to (b5.center);
\end{forest}
    \caption{Two examples of two-player extensive-form game trees with no chance moves and large information complexity $k$. In both examples, $k$ can be increased arbitrarily by increasing the branching factor of the root node. The left example would be easily reparable with a tighter definition of information complexity (that takes into account the fact that only one of the infosets in the second layer is reachable in any pure strategy profile), but the right example is not so easily reparable, and examples such as these are the reason that the proof of \Cref{th:pubchance} is more involved than one may initially expect.}\label{fig:pubchance-high-k}
\end{figure}

The rest of this subsection is devoted to proving this result, so, for the rest of this subsection, let $\Gamma$ be a two-player game with public chance actions, and call the two players \pone and \ptwo. Let $c$ be any of the three solution concepts. %
For a node $h^\tau$ in $\Gamma^c$, we will use $\tilde\sigma_i(h^\tau)$ to denote the sequence infosets reached and recommendations received by player $i$ has received on the path from the root to $h^\tau$, not including at $h$ itself. $\tilde\sigma_i(h^\tau)$ is always a valid player $i$ sequence in $\Gamma$. However, it is not the same as player $i$'s sequence $\sigma_i(h^\tau)$: for example, for NFCCE, if player $i$ deviated at $h^\tau$ then $\tilde\sigma_i(h^\tau) = \Root_i$ (because a deviating player receives no recommendations) but player $i$ still sees information sets and actions on the path to $h$.

Throughout this proof, for notational shorthand, we write $h^\tau \in I$, where $I$ is an infoset of player $i$ in $\Gamma$, if $h \in I$ and $\tau$ is not a trigger of player $i$.
\begin{lemma}\label{lem:pubchance1}
    Let $B$ be a mediator belief in $\Gamma^c$, and suppose (WLOG) that the mediator is giving a recommendation to player $1$. Then there exists a unique information set $I \in \mc I_\pone$, and a sequence $\sigma \in \Sigma_\ptwo$, such that:
    \begin{itemize}
        \item the mediator only gives a recommendation at information set $I$: for every $h^\tau \in B$, either $h \in I$ or $\tau$ is a trigger of \pone;
        \item for every $h^\tau \in B \cap I$, we have $\tilde \sigma_\ptwo(h^\tau) \preceq \sigma$; and
        \item there is an $h^\tau \in B \cap I$ with $\tilde \sigma_\ptwo(h^\tau) = \sigma$.
    \end{itemize}
Further, the map $B \mapsto (I, \sigma)$ is injective.
\end{lemma}
\begin{proof}
    Let $h_1^{\tau_1} \in B$ be any decision point for the mediator, and let $I \ni h_1$. 
    
    We first claim that there is no other node $h_n^{\tau_n}\in I' \ne I$ and player $1$ not having deviated. (See \Cref{fig:pubchance1} for a visual representation of the argument in this paragraph.) Let $h_1^{\tau_1}\text{---}h_2^{\tau_2}\text{---}\cdots\text{---}h_n^{\tau_n}$ be a path through the induced connectivity graph $G^c_\mediator[B]$. Further, assume WLOG that $h_2^{\tau_2} \notin I$ and $h_{n-1}^{\tau_{n-1}} \notin I'$; otherwise, move $h_1^{\tau_1}$ and $h_n^{\tau_n}$ along the path toward each other until this is true. We first ask: what is $h_1 \land h_n$? By definition of public chance, it cannot be a chance node, or else $h_1$ and $h_n$ could not be in the same public state, much less the same belief. It cannot be a \pone-node, because then the mediator cannot recommend \pone to play to both $h_1$ and $h_n$. It must therefore be a \ptwo-node. We now ask: how are $h_1^{\tau_1}$ and $h_2^{\tau_2}$ connected? $h_1^{\tau_1} \in I$ and $h_2^{\tau_2} \notin I$; therefore, $h_1$ and $h_2$ must be connected by an infoset at which the mediator recommends to \ptwo. Therefore, the mediator must recommend \ptwo to play to $h_1$. The same applies to $h_n^{\tau_n}$. But this is a contradiction, because it implies that $\ptwo$ must have been recommended two distinct actions at $h_1 \land h_n$.

    Now let $h_2^{\tau_2} \in B$ with $h_2 \in I$. We claim that either $\tilde\sigma_\ptwo(h_1^{\tau_1}) \preceq \tilde\sigma_\ptwo(h_2^{\tau_2})$ or $\tilde\sigma_\ptwo(h_1^{\tau_1}) \succeq \tilde\sigma_\ptwo(h_2^{\tau_2})$. Consider the node $h := h_1 \land h_2$. Since $h_1, h_2 \in I$ and $I$ belongs to \pone, $h$ must be a \ptwo-node (again, it cannot be a chance node, because chance is public). Let $a_1$ and $a_2$ be the actions at $h$ leading to $h_1$ and $h_2$ respectively. There are two cases: 
    \begin{itemize}
        \item The mediator does not recommend \ptwo to play to $h$, or recommends an action at $h$ that is neither $a_1$ nor $a_2$. Then $\tau_1 = \tau_2 = \sigma_\ptwo(h_1^{\tau_1}) = \sigma_\ptwo(h_1^{\tau_2})$.
        \item At $h$, the mediator recommends one of $a_1$ or $a_2$ (WLOG, $a_1$). Then $\tilde\sigma_\ptwo(h_2^{\tau_2}) = I(h) a_1 \preceq \tilde\sigma_\ptwo(h_1^{\tau_1})$.
    \end{itemize}
    Therefore, the set $\{ \tilde\sigma_2(h^\tau) : h^\tau \in B \cap I \}$ is totally ordered, and so it has a maximum element, which we call $\sigma$. Then, by definition, $\sigma$ satisfies the desired properties. 
    
    We therefore have a map $\phi : \mc B_\mediator^c \to (\mc I_1 \times \Sigma_2) \sqcup (\mc I_2 \times \Sigma_1)$ associating each mediator belief to a pair consisting of an infoset of one player and a sequence of the other player. 
    
    It remains to show that $\phi$ is injective. Let $(I, \sigma) \in \mc I_1 \times \Sigma_2$ (WLOG). Let $h^\tau \in I$ with $\sigma_\ptwo(h^\tau) = \sigma$, and pick an $h^\tau$ so that $\tau = \bot$ if one exists. First, suppose $\tau = \bot$. There is only one way to reach $h^\bot$: at every belief $B'$, the mediator must play the action leading to $h$, and then observe the public observation containing $h'$. Thus, the belief containing $h^\bot$ must be unique.

If $B \cap I$ contains no trigger-less node, then it contains only nodes with \ptwo-triggers. But then $B \subseteq I$, because no node with a \ptwo-trigger can ever be connected to a node with a \pone-trigger, and every node in $B \setminus I$ must have a \pone-trigger because of \Cref{lem:pubchance1}. But this precisely fixes what $B$ is: namely, $B = \{ h^\tau \in I : \sigma_\ptwo(h^\tau) \preceq \sigma \}$, because all such $h^\tau$ must be in $B$, and \Cref{lem:pubchance1} states that no others can be.
\end{proof}

Thus, the number of beliefs is polynomial in the size of the game. Since every belief overlaps exactly one mediator information set, it follows that $R^c_\mediator$ is polynomial in the game size. This completes the proof of \Cref{th:pubchance}.

\begin{figure}[t]
\centering
\tikzset{
    every node/.style={},
    infoset/.style={-, densely dotted, ultra thick, color=p2color},
}
\forestset{
    default preamble={for tree={parent anchor=south, child anchor=north}},
    p1/.style={regular
        polygon, regular polygon
        sides=3, inner sep=2pt, fill=p1color, draw=none},
    p2/.style={p1, shape border rotate=180, fill=p2color},
    pn/.style={draw, circle, inner sep=2pt},
}
\begin{forest}
[,p2,s sep=1cm
    [,p1,edge label={node[left]{$h_1$}}
        [,p2,name=2]
    ]
    [,p1,no edge,edge label={node[right]{$h_2$}}
        [,p2,name=1]
    ]
    [$\cdots$,no edge]
    [,p1,no edge,edge label={node[left]{$h_{n-1}$}}
        [,p2,name=3]
    ]
    [,p1,edge label={node[right]{$h_n$}}
        [,p2,name=4]
    ]
]
\draw[infoset] (1) -- (2); 
\draw[infoset] (3) -- (4); 
\end{forest}
\caption{Visualization of the proof of \Cref{lem:pubchance1} (other nodes, such as the ancestors of $h_2$, are not shown). Since both \ptwo infosets are used to make connections in $G^c[B]$, they must both be played to by \ptwo, which is impossible since this would require \ptwo to make two different moves at the top node.}\label{fig:pubchance1}
\end{figure}

\subsection{Discussion: Relationship to Triangle-Freeness}\label{se:pubchance-discussion}

\Cref{th:pubchance} implies that \Cref{al:dag} runs in polynomial time in two-player games of public chance. As we mentioned, we are not the first to exhibit a polynomial-time algorithm in this setting; \citet{Farina20:Polynomial} has exhibited one using a different technique, namely by showing that the {\em von Stengel--Forges} (vSF) polytope~\cite{Stengel08:Extensive} is tight. It is instructive to compare the two approaches. The approach of \citet{Farina20:Polynomial} carries many similarities to our approach for this special case---in particular, their approach also works by effectively constructing a DAG representation of $\Xi^\efce$. However, while their approach dynamically chooses which information set to expand next on the fly, our approach uses the fixed ordering provided by the timeable game to decide which information set is ``next''.  When the game is timeable, our approaches give essentially the same representation: indeed, the proof in the previous section shows that there is a decision point of the mediator in $\Gamma^c$ for every relevant pair $(I_1, \sigma_2)$ or $(\sigma_1, I_2)$, which are precisely the branching points in the representation of \citet{Farina20:Polynomial}.

\begin{figure}
\begin{center}
\tikzset{
    every node/.style={},
    infoset/.style={-, densely dotted, ultra thick, color=p2color},
}
\forestset{
    default preamble={for tree={parent anchor=south, child anchor=north}},
    p1/.style={regular
        polygon, regular polygon
        sides=3, inner sep=2pt, fill=p1color, draw=none},
    p2/.style={p1, shape border rotate=180, fill=p2color},
    pn/.style={draw, circle, inner sep=2pt},
            nat/.style={draw},
        terminal/.style={draw=none, font=\scriptsize\sf, inner sep=2pt},
        dum/.style={inner sep=0pt, minimum size=0.1px, circle, fill=black},
}
\begin{forest}
[,nat
[,p1 [] [,p2,l+=5mm,s=1cm,name=a [] []]]
[,p1 [] [,p2,l+=5mm,s=1cm [] []]]
[,p1 [] [,p2,l+=5mm,s=1cm,name=b [] []]]
]
\draw[infoset] (a) -- (b);
\end{forest}
\end{center}
\caption{An example of a timeable triangle-free game in which our construction will be exponentially-sized (in the branching factor of the root node). In this game, the algorithm of \citet{Farina20:Polynomial} works by essentially ``re-ordering'' the game tree so that \ptwo's decision point is moved to the root, at which point the chance decision can be treated as public, thereby removing the exponentiality.}\label{fig:trianglefree}
\end{figure}

Unlike their approach, our correlation DAG algorithm provides \revision{a parameterized} guarantee on any game. However, it is limited to timeable games, whereas theirs generalizes beyond timeable games to a family they coin {\em triangle-free games}. Here, for the sake of completeness, we include a definition of triangle-freeness.

\begin{definition}
    In a two-player game, two information sets $I_1 \in \mc I_1$ and $I_2 \in \mc I_2$ are {\em connected}, denoted $I_1 \rele I_2$, if there exists a node $h$ with $h \succeq I_1$ and $h \succeq I_2$. A {\em triangle} is a collection of four infosets $I_1, I_1' \in \mc I_1$ and $I_2, I_2' \in \mc I_2$ such that $I_1 \rele J_1$, $I_2 \rele J_2$, and $I_1 \rele J_2$.
\end{definition}

Intuitively, triangle-freeness is useful because it guarantees the existence of some ``branching order'' that can be used to fill in the polytope $\Xi^\efce$. We refer the reader to the paper of \citet{Farina20:Polynomial} for more details. It is not difficult to construct triangle-free games in which our construction would be exponentially-sized; see \Cref{fig:trianglefree}. We leave to future research the question of whether it is possible to extend our algorithm so that it is also runs in polynomial time in all triangle-free games, achieving the best of both worlds.

\subsection{Fixed-Parameter Hardness of Representing $\Xi^\efcce$ and $\Xi^\efce$}
A natural question is whether it is possible to achieve the same bound for EFCCE and EFCE as achieved for NFCCE and team games---namely, a construction whose exponential term depends only on $b$ and $k$. It turns out that our construction does {\em not} accomplish this, and in fact, {\em no} representation of $\Xi^c$ for $c = \efcce$ or $c = \efce$ can have size $O^*(f(k))$ for any function $f$ under standard complexity assumptions even when $b=2$. To do this, we first review some fundamental notions of {\em parameterized complexity}.
\begin{definition}
    A {\em fixed-parameter tractable} (FPT) algorithm for a problem is an algorithm that takes as input an instance $x$ and a {\em parameter} $k \in \N$, and runs in time $f(k) \poly(\abs{x})$, where $\abs{x}$ is the bit length of $x$ and $f : \mc \N \to \N$ is an {\em arbitrary function}.
\end{definition}

The $k$-CLIQUE problem\footnote{The $k$-CLIQUE problem is to decide whether a given graph contains a clique of size at least $k$.} is widely conjectured to not admit an FPT algorithm parameterized by the clique size $k$. In the literature on parameterized complexity, this conjecture is known as {\em FPT $\ne$ W[1]}, and is implied by the exponential time hypothesis~\cite{Chen05:Tight}. We now show that this conjecture implies lower bounds on the complexity of representing the polytopes $\Xi^\efcce$ and $\Xi^\efce$.
\begin{restatable}{theorem}{thmWOneHardness}\label{th:w1-hardness}
    Assuming FPT $\ne$ W[1], there is no FPT algorithm for linear optimization over $\Xi^\efcce$ or $\Xi^\efce$ parameterized by information complexity, even in two-player games with constant branching factor.
\end{restatable}

    \begin{figure}
        \centering
            \tikzset{
    every path/.style={-},
    every node/.style={draw},
    infoset/.style={-, densely dotted, ultra thick, color=p2color},
    terminal/.style={draw=none},
}
\forestset{
    l sep=8mm,
    default preamble={for tree={parent anchor=south, child anchor=north}},
    p1/.style={regular
        polygon, regular polygon
        sides=3, inner sep=2pt, calign=last, fill=p1color, draw=none},
    p2/.style={p1, shape border rotate=60, fill=p2color},
}
\begin{forest}
baseline
[
[,p1 [,terminal] [,p1 [,terminal] [,p1 [,terminal]
    [
        [,p2,name=1 [,terminal] [,p2,name=5 [,terminal] [,terminal] ]]
        [,p2,name=2 [,terminal] [,p2,name=6 [,terminal] [,terminal] ]]
    ]
]]]
[,p1 [,terminal] [,p1 [,terminal] [,p1 [,terminal]
    [
        [,p2,name=3 [,terminal] [,p2,name=7 [,terminal] [,terminal] ]]
        [,p2,name=4 [,terminal] [,p2,name=8 [,terminal] [,terminal] ]]
    ]
]]]
]
\draw[infoset, bend left] (1) to (3);
\draw[infoset, bend right] (2) to (4);
\draw[infoset, bend left] (5) to (7);
\draw[infoset, bend right] (6) to (8);
\end{forest}
        \caption{The game used in the proof of \Cref{th:w1-hardness}, for $n=k=2$.}\label{fig:w1-hard}
    \end{figure}
\noindent{\em Proof.}
    We reduce from $k$-CLIQUE. Let $G = (V, E)$ be a graph with $n$ nodes (identified with the positive integers $[n]$), and construct the following two-player game $\Gamma$ (see also \Cref{fig:w1-hard}):
    \begin{itemize}
        \item Chance chooses an integer $j_1 \in [k]$ and tells \pone but not \ptwo. Transition to the node $(j_1, 1)$.
        \item For each $v_1 \in [n+1]$, the node $(j_1, v_1)$ is a decision node for \pone. \pone may {\em exit} or {\em continue}. If \pone exits, transition to the terminal node $(j_1, v_1, {\sf E})$. Otherwise, transition to $(j_1, v_1+1)$.
        \item At the node $(j_1, n+2)$, Chance chooses an integer $j_2 \in [k]$ and tells \ptwo, Transition to the node $(j_1, j_2, {\sf E})$.
        \item For each $v_2 \in [n]$, the node $(j_1, j_2, v_2)$ is a decision node for \ptwo. \ptwo may {\em exit} or {\em continue}. If \ptwo exits, transition to the terminal node $(j_1, j_2, v_2, {\sf E})$. Otherwise, transition to  $(j_1, v_1+1)$.
        \item Finally, $(j_1, j_2, n+1)$ is a terminal node for all $j_1, j_2$.
    \end{itemize}
    Since this result is only concerned with representing the correlation plan polytope (not necessarily with computing optimal equilibria), we do not need to specify utilities or chance probabilities---these do not affect the construction of the augmented game $\Gamma^c$ nor the polytope $\Xi^c$.\footnotemark{}
    We will identify the information sets of both players $i$ by $(j_i, v_i)$ for $j \in [k]$, and the infoset-action pairs by $(j_i, v_i, {\sf E})$ and  $(j_i, v_i, {\sf C})$ for exiting and continuing respectively.

    $\Gamma$ has information complexity $2k$ since every public state has at most $k$ sequences for each player. Every non-chance node has branching factor exactly $2$. 
    \footnotetext{Note that $\Xi^c$ is not the set of EFCEs---it is a representation of the set of correlation plans. That set does not depend on utilities or chance probabilities.}
    
Given a correlation plan $\vec\xi$, define the vector $\vec m^{\vec\xi} \in [0, 1]^{[k] \times [n] \times [k] \times [n]}$ where $\vec m^{\vec\xi}[j_1, v_1, j_2, v_2]$ is the probability that each player $i$ exits at exactly the $v_i$th opportunity conditioned on observing $j_i$. 
Notice that, for $j_1, j_2 \in [k]$ and $v_1, v_2 \in [n]$,  $\vec m^{\vec\xi}[j_1, v_1, j_2, v_2]$ is a linear function of both the correlation plan spaces  $\Xi^\efcce$ and $\Xi^\efce$: for $\vec\xi \in \Xi^\efcce$, it is exposed as $\xi[(j_1, j_2, v_2, {\sf E})^{(j_1, v_1+1)}]-\xi[(j_1, j_2, v_2, {\sf E})^{(j_1, v_1)}]$; for $\vec\xi \in \Xi^\efce$, it is exposed as $\xi[(j_1, j_2, v_2, {\sf E})^{(j_1, v_1, {\sf E})}]$. (For $\Xi^\nfcce$, $\vec m^{\vec\xi}$ is not a linear function of $\vec\xi$, so, as expected, the argument fails here.)

    Let $M = \{ \vec m^{\vec\xi} : \vec\xi \in \Xi^c\} \subseteq [0, 1]^{[k] \times [n] \times [k] \times [n]}$ be the polytope of vectors $\vec m$ corresponding to correlated strategies. At this point, since $M$ does not depend on the notion of equilibrium, we have no more need to distinguish between EFCCE and EFCE. It suffices to show that linear optimization on $M$ can decide $k$-CLIQUE. First, we characterize the vertices of $M$. A vertex of $M$ is characterized by, for each player $i \in \{1, 2\}$ and each $j \in [k]$, picking at most one vertex $v_{i,j} \in [n]$, and constructing $\vec m$ by setting $\vec m[j_1, v_1, j_2, v_2] = {\vec 1} \qty{ v_{1,j_1} = v_1 \text{ and } v_{2, j_2} = v_2}$. Now consider the objective function $g : M \to \R$ defined by
    \begin{align}
        g(\vec m) = \E_{\substack{j_1, j_2 \in [k] \\ v_1, v_2 \in [n]}} 
        \begin{cases}
            \vec m[j_1, v_1, j_2, v_2] & \text{if $j_1 = j_2 \text{ and } v_1 = v_2 \le n; \text{  or } j_1 \ne j_2 \text{ and } (v_1, v_2) \in E$} \\
            0                     & \text{otherwise}
        \end{cases}
    \end{align}
    where the expectation is over a uniformly random sample. We now claim that $\max_{\vec m \in M} g(\vec m) = 1$ if and only if $G$ has a clique of size $k$, which will complete the proof.
    \begin{itemize}
        \item [($\Leftarrow$)] If $G$ has a $k$-clique $\{ v^*_1, \dots, v^*_k \}$, then we set $v_{i, j} = v^*_j$ for both players $i \in \{1, 2\}$, and indeed this achieves $g(\vec m) = 1$ by construction.
        \item [($\Rightarrow$)] If $g(\vec m) = 1$, then for all $j$ we must have $\sum_{v \in [n]} \vec m[j, v, j, v] = 1$, \textit{i.e.}, $v_{1, j} = v_{2, j}$. But then $\{ v_{1, j}, \dots, v_{1, k}\}$ must be a clique by construction, because otherwise there would be some $j_1 \ne j_2$ for which $\vec m[j_1, v_{1, j}, j_2, v_{1, j}] = 0$.  \hfill\qed
    \end{itemize}

Technically speaking, this result does not establish parameterized hardness of computing optimal EFCCEs or EFCEs, as there could hypothetically be a method for doing so that exploits the special nature of the~\eqref{eq:lp}. Indeed, the proof of \Cref{th:w1-hardness} exploits the fact that the objective coefficient $g[h^\tau]$ may depend on $\tau$ as well as $h$, which is not the case for the LP \eqref{eq:lp}. However, we know of no technique for optimal equilibria that would not also imply the ability to optimize over $\Xi^c$. Therefore, \Cref{th:w1-hardness} is a lower bound that applies to all known techniques for computing optimal EFCCEs and EFCEs.

\section{Two-Sided Column Generation Approach}\label{se:colgen}

The approach in the previous section overcomes the imperfect recall of the mediator player in the augmented game by using an extended formulation, in which the imperfect recall is eliminated at the cost of a (controlled) exponential increase in the size of the decision problem. As we have seen, this process begets new parameterized complexity results for the problem of computing optimal correlated solution concepts in extensive-form games, as well as, in many cases, the current state-of-the-art algorithm for computing optimal correlated solution concepts in EFGs. Yet, there are some cases in which the exponential reformulation is prohibitive.
In this section, we propose an alternative approach based on column generation (with some domain-specific tweaks and improvements) that might be helpful in such scenarios. To reduce the burden on notation, our presentation will focus on the case of two players, though in principle the method can be extended to multiple players without significant obstacles.

\subsection{Semi-Randomized Correlation Plans}\label{sec:semi randomized}

We introduce the following notation for two-player games, which was also used by \citet{Farina21:Connecting}: it is worth re-emphasizing that sequence pair $(I_1 a_1, I_2 a_2) \in \Sigma_1 \times \Sigma_2$ is {\em relevant} if there is a history $h$ with $h \succeq I_1$ and $h \succeq I_2$ (see also \Cref{def:connected-infosets}). We write $\sigma_1 \rele \sigma_2$ to denote a relevant sequence pair, and $\Sigma \subseteq \Sigma_1 \times \Sigma_2$ for the set of all relevant sequence pairs. For a sequence $\sigma$ of Player 1 and an infoset $I$ of Player 2, we write $\sigma \rele I$ if $\sigma \rele Ia$ for each action $a$. Similarly, for a Player 1 infoset $I$ and Player 2 sequence $\sigma$, we write $I \rele \sigma$ if $Ia \rele \sigma$. In the subsequent discussion, for ease of notation, the symbol $\Xi$ will be used to denote $\Xi^\efce$, as it encompasses both the EFCCE and NFCCE cases.

Now, we introduce the strategy representation which we employ in our algorithm. We observe that variables in LP~\eqref{eq:lp}  belong to the convex polytope $\Xi$, but that polytope cannot be compactly represented in general. 
Therefore, we tackle LP~\eqref{eq:lp} by adopting the notion of \emph{semi-randomized correlation plan} proposed by \citet{Farina21:Connecting}. For completeness, we show how \emph{semi-randomized correlation plan} can be derived from the \emph{von Stengel-Forges polytope} \cite{Stengel08:Extensive} representing interlaced sequence-form ``probability mass conservation'' constraints for the two players.
\begin{definition}\label{def:vsf}
    The \emph{von Stengel-Forges polytope}, denoted $\mc V$, is the polytope of all vectors $\vec{\xi} \in \mathbb{R}_{\ge 0}^{\Sigma}$ (\textit{i.e.}, indexed over relevant sequence pairs) such that: 
    \begin{align}
        &\circled{A}\quad \vec \xi[\Root_1,\Root_2] = 1 
        \\& \circled{B}\quad \sum_{a \in A_{I}} \vec \xi[Ia,\sigma_{2}] = \vec\xi[\sigma_1(I),\sigma_{2}] \quad \forall I \rele \sigma_2\in \mc I_1\times\Sigma_2, \qq{and} 
        \\&\circled{C}\quad \sum_{a \in A_{I}} \vec\xi[\sigma_1,Ia] = \vec\xi[\sigma_1,\sigma_{2}(I)] \quad \forall I \rele\sigma_2\in \mc I_2\times\Sigma_1.
    \end{align}
\end{definition}

The set of linear constraints defining $\mc V$ is polynomially-sized. Moreover, the set of correlation plans is a subset of the von Stengel-Forges polytope, that is, $\Xi \subseteq \mc V$ \cite{Stengel08:Extensive}. 

The notation in this section is different from that in the preceding section: $\vec\xi$ here is indexed by {\em relevant sequence pairs}, whereas $\vec\xi$ in \Cref{se:cp} is indexed by {\em sequences for the mediator in the augmented game}. However, the two notations essentially describe the same thing: given a history $h^\tau$ of any of the augmented games, we identify the realization-form player reach probability $\vec\xi[h^\tau]$ with the term $\vec\xi[\sigma_1(h^\tau), \sigma_2(h^\tau)]$, where $\sigma_i(h^\tau) = \sigma_i(\tau)$ if $\tau$ is a trigger of player $i$, and otherwise $\sigma_i(h)$ (that is, $\sigma_i(h^\tau)$ is the sequence of recommendations that the mediator has given to player $i$ on the path to node $h^\tau$ in the augmented game.)

Finally, a \emph{semi-randomized correlation plan} is composed of a deterministic sequence form strategy for one player, while the other player independently plays a mixed strategy\footnote{The ideas in this section also naturally extend to games with more than two players, where a semi-randomized correlation plan is a correlation plan in which a single player is allowed to randomize and the other players must play pure strategies. The notation for that would be significantly more cumbersome, and the performance benefit would be less noticeable because the majority of players would still be forced to play pure strategies; as such, we restrict our attention in this section to the two-player case.}. 
\begin{definition}[\cite{Farina21:Connecting}]
The sets of semi-randomized correlation plans are 
\[
\Sr[1]\defeq \mleft\{ \vec\xi\in\mc V: \vec\xi[\Root, \sigma_{2}] \in \{0,1\} \,\, \forall\ \sigma_{2} \in \Sigma_{2}\mright\}\,\,\text{\normalfont and }\,\,
\Sr[2]\defeq \mleft\{ \vec\xi\in\mc V: \vec\xi[\sigma_1, \Root] \in \{0,1\} \,\, \forall\ \sigma_{1} \in \Sigma_{1}\mright\}.
\]
\end{definition}

Given $i\in\{1,2\}$, a point $\vec\xi\in\Sr$ can be expressed using real and binary variables, in addition to the linear constraints defining the von Stengel-Forges polytope $\mc V$. In particular, we rely on the observation by \citet{Farina21:Connecting} that $\Xi=\co(\Sr[1])=\co(\Sr[2])=\co(\Sr[1]\cup\Sr[2]).$

\subsection{Correlation-Plan Decomposition and Iterative Framework}

We say that a correlation plan $\vec\xi$ is a \emph{product correlation plan} if, for any $(\sigma_1,\sigma_2)\in\Sigma$, $\vec\xi[\sigma_1,\sigma_2] = \vec\xi[\sigma_1,\Root]\cdot\vec\xi[\Root,\sigma_2]$.
Since any semi-randomized correlation plan corresponds to a distribution of play where one player plays a pure sequence-form strategy, while the other plays a mixed sequence-form strategy, $\vec\xi\in\Sr$ is guaranteed to be a product correlation plan for any $i$ (see \cite[Lemma 3]{Farina21:Connecting}).

Given $\vec\xi\in \Sr[1]\cup\Sr[2]$, let $\marg[1]\in \X_1$ be the \emph{marginal vector} such that
$\margi[1][\cdot]\defeq\vec\xi[\cdot,\Root]$, and let $\marg[2]$ be defined analogously. %
Then, we can decompose any correlation plan $\vec\xi\in\Sr[1]\cup\Sr[2]$ as 
\[
\vec\xi=\lambda\, \vec x_1\marg[2]^\top + (1-\lambda)\,\marg[1]\vec x_2^\top,
\]
for some appropriate choice of $\lambda \in \{0, 1\}$, and mixed strategies $\vec x_1,\vec x_2$ for Player 1 and Player 2, respectively\footnote{Notation like $\vec x_1 \marg[2]^\top$ is technically an abuse of notation, since $\vec\xi$ is only indexed over {\em relevant} sequence pairs $\Sigma \subset \Sigma_1 \times \Sigma_2$. To be fully precise, by writing $\vec a \vec b^\top \in \R^\Sigma$ we refer to the vector $(\vec a \vec b^\top)[\sigma_1, \sigma_2] = \vec a[\sigma_1] \vec b[\sigma_2]$. But we will still call it an ``outer product''.}. Moreover, given $\vec\xi\in\Sr[1]\cup\Sr[2]$, we have $\lambda\, \vec x_1\marg[2]^\top + (1-\lambda)\,\marg[1]\vec x_2^\top \in\co(\Sr[1]\cup\Sr[2])$ for any $\lambda\ge0$ and well-formed sequence-form strategies $\vec x_1,\vec x_2$.

Our column-generation algorithm will solve a sequence of linear programs. Each program refines the previous by allowing the correlation plan variable $\vec\xi$ to be expressed as a convex combination of points from a growing {\em support set} $S$. In our case, the support generated after any $T$ iterations is in the form $S = \{ \marg[1]^{(1)}, \marg[2]^{(1)}, \marg[1]^{(2)}, \marg[2]^{(2)},  \dots, \marg[1]^{(T)}, \marg[2]^{(T)} \}$, where $\marg[1]^{(t)}$ and $\marg[2]^{(t)}$ are marginal strategies for P1 and P2, respectively. With this support, we will allow the mediator to select any {\em mixture of semi-randomized correlation plans}, where {\em at least one} of the players is playing a strategy in $S$. Formally, the mediator is allowed to select weights $\lambda_i^{(t)}$, for $i \in \{1, 2\}$ and $t \in [T]$, such that $\sum_{i=1}^2 \sum_{ t=1}^T \lambda_i^{(t)} = 1$ and $\lambda_i^{(t)} \ge 0$, as well as the other player's strategy $\vec x_{i}^{(t)} \in \X_i$, resulting in the correlation plan
\begin{align}
    \vec\xi := \sum_{ t=1}^T \lambda_1^{(t)} \qty((\marg[1]^{(t)})(\vec x_2^{(t)})^\top + \lambda_2^{(t)} (\vec x_1^{(t)})(\marg[2]^{(t)})^\top).\label{eq:sr-xi}
\end{align}
We will denote by $\Xi^S$ the set of all correlation plans admissible in the above sense for a given support $S$.

Compared to \cite{Farina21:Connecting}, this notion of mixing allows more correlation plans to be formed for any given support size $T$: that paper {\em fixes} upfront which player is allowed to play a mixed strategy and which player plays the pure strategy, whereas ours allows the master problem to decide this. We say that the column-generation algorithm of \citet{Farina21:Connecting} is \emph{one-sided} since one player has to select a pure strategy, while the other picks a sequence-form strategy after observing the pure strategy. In contrast, we call our framework \emph{two-sided}, each player can have both roles, and the parameter $\vec\lambda$ dictates who has which role.
As such, the correlation-plan decomposition which we introduced allows us to exploit correlation plans already in the support $S$ in a more powerful way than what is possible in other \emph{one-sided} column-generation approaches like the one by~\citet{Farina21:Connecting}. In particular, we remark that $\tilde{\vec x}_{1}^{(t)}, \tilde{\vec x}_{2}^{(t)}$ are continuous variables in LP~\eqref{eq:master}. Therefore, each player is allowed to re-optimize their mixed strategies, enabling them to reach a richer set of correlation plans starting from the same support set. As a result, for any $T$, our master problem will be substantially tighter, leading to faster convergence.

Given an equilibrium concept $c$, our master problem at time $T$ is the following linear program:

\begin{equation}\label{eq:master}
\left\{\begin{aligned}
            \max\quad & \vec g^\top\vec\xi \\[2mm]
            \text{s.t.}\quad  &\circled{1}~ \vec A^\top_i \vec\xi \le (\vec F_i^c)^\top \vec v_i&& \forall i \in \{1, 2\} \\
            &\circled{2}~ (\vec f_i^c)^\top \vec v_i \le \vec b_i^\top \vec\xi&& \forall i \in \{1, 2\}\\
            &\circled{3}~\vec\xi = \sum_{ t=1}^T \qty((\marg[1]^{(t)})(\tilde {\vec x}_2^{(t)})^\top + (\tilde {\vec x}_1^{(t)})(\marg[2]^{(t)})^\top) \\
            &\circled{4}~  \vec F_i\tilde{\vec x}_{i}^{(t)}=\lambda_{i}^{(t)}\vec f_i, \quad \tilde {\vec x}_i^{(t)} \ge 0 &&\forall i\in\{1,2\}, t\in [T]\\
            &\circled{5}~  \sum_{i=1}^2\sum_{t=1}^T \lambda_i^{(t)}=1 , \quad \lambda_i^{(t)}\ge 0~\forall (i,t)
\end{aligned}\right.\tag{M}
\end{equation}
where $\vec g, \vec F_i^c, \vec f_i^c, \vec v_i, \vec A_i$ are defined in \eqref{eq:lp}, and $\vec F_i = \vec x_i$ are the sequence-form equality constraints for player $i$ in the original game $\Gamma$---that is, $\X_i := \{ \vec x_i \ge \vec 0 : \vec F_i \vec x_i = \vec f_i \}$.

\begin{itemize}
    \item each correlation plan $\vec\xi$ is represented through the decomposition which we defined above, where we perform the change of variables $\tilde{\vec x}_i^{(t)}\defeq \lambda_i^{(t)}\,\vec{x}^{(t)}_i$ for each $i$;
    
   \item Constraint \circled{1} and \circled{2} are the same constraints as in \eqref{eq:lp}.
   \item Constraints~\circled{3} to~\circled{5} define $\vec \xi$ according to the aforementioned semi-randomized mixture. They take the place of the hard constraint~\circled{$\star$} in \eqref{eq:lp}.
\end{itemize}

If we could afford to set $S=\Pi_1$ or $S = \Pi_2$, finding an optimal NFCCE, EFCCE, EFCE for an arbitrary objective $\vec{g}$ would amount to solving LP~\eqref{eq:master} once. However, $\Pi_1$ and $\Pi_2$ are usually exponentially large. 
Therefore, we follow the approach by~\citet{Ford58:Suggested} and generate the support $S$ iteratively.

Algorithm~\ref{al:colgen} describes the main steps of our iterative procedure. First, we initialize the support $S$ through a seeding phase in which $S$ is endowed with one or more correlation plans which are known to belong to $\Xi$. In our experiments, we start by assigning to $S$ the correlation plan obtained as the product of one \emph{uniform} mixed strategy per player (\textit{i.e.}, a strategy such that, at each $I$, the player draws one action from $A_I$ according to a uniform probability distribution). Then, at each iteration $t$, we solve the master LP \eqref{eq:master} with the current support $S$. Each time we solve \eqref{eq:master}, we keep track of the resulting primal and dual variables. In particular, when solving \eqref{eq:master}, the algorithm keeps track of the current solution $\vec\xi^{(T)}$ (\textit{i.e.}, the correlation plan corresponding to the optimal decomposition), and the dual variables for constraints \circled{1} and \circled{2}. %

\subsection{Expansion of Support: The Pricing Problem}

At iteration $T$, the marginal strategies $\vec m_1^{(T+1)}, \vec m_2^{(T+1)}$ of a new correlation plan $\vec{\xi}^{(T+1)}$ are added to $S$. The selection of $\vec\xi^{(T+1)}$ is controlled by the function \textsc{Pricer} (\Cref{al:colgen}), which solves the pricing problem of finding the correlation plan that would lead to the maximum gradient of the objective (\textit{i.e.}, maximum \emph{reduced cost}) if it was to be included in the convex combination computed by \eqref{eq:master}. At iteration $t$, such correlation plan can be computed from the solution to the dual of the master LP. We use a tuple $((\tilde{\vec x}_i^c)^{(T)}, \gamma_i^{(T)})_{i \in [n]}$, where $(\tilde{\vec x}_i^c)^{(T)}$ be the sub-vector of the dual variables corresponding to constraints \circled{1} and player $i$, and $\gamma_i^{(T)}$ be the sub-vector of dual variables corresponding to constraints \circled{2} and trigger $\tau$. Then, by letting
$$
\vec w^{(T)} \defeq \sum_{i \in \{1, 2\}}\qty(\vec A_i(\tilde{\vec x}_i^c)^{(T)} -\gamma_i^{(T)}\vec b_i),
$$
the pricing problem amounts to solving $\max_{\vec\xi\in\Xi}(\vec g-\vec w^{(T)})^\top  \vec \xi$.
We know that $\Xi=\co(\Sr[i])$, for any $i$. Therefore, by linearity of the objective and by convexity, we have 
\[\max_{\vec\xi\in\Xi}~(\vec g-\vec w^{(T)})^\top \vec \xi=\max_{\vec\xi\in\Sr}~(\vec g-\vec w^{(T)})^\top \vec \xi.\] 
This is a well-defined mixed integer LP (MIP), which can be solved through a commercial solver such as Gurobi. We denote by $\delta$ the optimal value of the pricing problem, and by $\vec\xi^{(t+1)}$ a correlation plan attaining such value (see Line~\ref{line:pricer}).

In the initial stage of the algorithm, \eqref{eq:master} may be infeasible as the support generated up to that point might be insufficient to generate an equilibrium. In this case, we define an LP (M') by replacing constraint \circled{2} with $(\vec f_i^c)^\top \vec v_i \le \vec b_i^\top \vec\xi+u$, where $u$ is an introduced slack variable, and making the objective $\max (-u)$. This LP is guaranteed to be feasible. 

The column generation algorithm admits a clear game-theoretic interpretation\footnote{In this paragraph only, we drop the superscript $(T)$ for simplicity of notation.}. The dual variables $\tilde{\vec x}_i^c$ correspond to an augmented-game deviator strategy ${\vec x}_i^c := \tilde{\vec x}_i^c/\gamma_i \in \X_i^c$, scaled by a Lagrangian multiplier $\gamma_i \ge 0$. In the initial (infeasible) phase, it solves the zero-sum game
\begin{align}
    \max_{\vec\xi \in \Xi^S} \min_{\substack{i \in [n],\\ \vec x_i^c \in \X_i^c}} {-\vec\xi^\top\qty( \vec A_i \vec x_i^c - \vec b_{\tau})}\label{eq:infeas}
\end{align}
by column generation for the maximizing player. For any given $\vec\xi$, the objective value above is nothing but the maximum deviation benefit of a deviator strategy $\vec x_i^c$ against $\vec\xi$. Therefore, since every game {\em has} equilibria, the equilibrium value of this game is zero, and that value is achieved when $S$ is large enough that $\Xi^S$ contains at least one equilibrium, at which point the algorithm moves to the second stage. In the second stage, the algorithm solves the zero-sum game 
\begin{align}
    \max_{\vec\xi \in \Xi^S} \min_{\substack{\gamma_i \ge 0,\\ \vec x_i^c \in \X_i^c}} \vec g^\top \vec\xi - \sum_{i \in \{1, 2 \}} \gamma_i \cdot \vec\xi^\top\qty( \vec A_i \vec x_i^c - \vec b_{\tau}) \label{eq:feas}
\end{align}
again by column generation. The optimal solutions to this game are, by construction, the optimal equilibria. The feasibility of \eqref{eq:infeas} guarantees that \eqref{eq:feas} has a finite equilibrium value, and moreover the equilibria of this zero-sum game are precisely the solutions to \eqref{eq:lp}, that is, the optimal equilibria with concept $c$. In the language of this zero-sum game, the pricer is the maximizing player (mediator)'s best response value against the current equilibrium $(\vec\xi,\{{\vec x}_i^c,\gamma_i\}_{i\in\{1,2\}}$. Thus, when equilibrium is achieved (\textit{i.e.}, when the pricer fails to find an improving direction), the game is solved.

\begin{algorithm}[H]
\caption{Two-Sided Column Generation}
\label{al:colgen}
\DefineFunction{ComputeOpt}
\DefineFunction{Pricer}
\Function{\ComputeOpt{game $\Gamma$, concept $c\in\{\nfcce,\efcce,\efce\}$, objective $\vec g$, tolerance $\epsilon \ge 0$}}{
\textbf{initialization:} $T \gets 1, S \gets \{\text{uniformly random strategy for both players}\}$\;
\While{{\normalfont within computational budget}}{
    \If{\eqref{eq:master} {\normalfont is infeasible}}{
         (M') $\gets$ replace constraint \circled{2} in \eqref{eq:master} with $(\vec f_i^c)^\top \vec v_i \le \vec b_i^\top \vec\xi+u$, where $u$ is an introduced slack variable, and replace the objective with $\max (-u)$\Comment*{\color{commentcolor}see discussion in text above]}
         $\vec\xi^{(T)},\{(\tilde{\vec x}_i^c)^{(T)},\gamma_i^{(T)}\}_{i\in\{1,2\}}\gets$ solve (M') \Comment*{\color{commentcolor}primal $\curxi$, dual $\{(\tilde{\vec x}_i^c)^{(T)},\gamma_i^{(T)}\}_{i\in\{1,2\}}$]}
    $\beta,\vec m_1^{(T+1)}, \vec m_2^{(T+1)}\gets\textsc{Pricer}(\vec \xi^{(T)}, \vec 0, \{(\tilde{\vec x}_i^c)^{(T)},\gamma_i^{(T)}\}_{i\in\{1,2\}})$\;
    }
        $S\gets S\cup\{\vec m_1^{(T+1)}, \vec m_2^{(T+1)}\}$\;
    $T \gets T+1$\;
    \Else {
    $\vec\xi^{(T)},\{(\tilde{\vec x}_i^c)^{(T)},\gamma_i^{(T)}\}_{i\in\{1,2\}}\gets$ solve \eqref{eq:master} \Comment*{\color{commentcolor}primal $\curxi$, dual $\{(\tilde{\vec x}_i^c)^{(T)},\gamma_i^{(T)}\}_{i\in\{1,2\}}$]}
        $\beta,\vec m_1^{(T+1)}, \vec m_2^{(T+1)}\gets\textsc{Pricer}(\vec \xi^{(T)}, \vec g, \{(\tilde{\vec x}_i^c)^{(T)},\gamma_i^{(T)}\}_{i\in\{1,2\}})$\label{line:pricer}\;
        \lIf{$\beta \le \epsilon$}{\Return $\vec\xi^{(T)}$}
}
    }
}
\Hline{}
\Function{\Pricer{$\vec\xi, \vec g, \{\tilde{\vec x}_i^c,\gamma_i\}_{i\in\{1,2\}}$}}{
$\vec{w}\gets\sum_{i\in\{1,2\}}\qty(\vec A_i\tilde{\vec x}_i^c-\gamma_i\vec b_{i})$\;
Select one player $i$\;
$\beta, \vec\xi^* \gets \max_{\vec\xi^*\in\Sr}~ (\vec g - \vec w)^\top (\vec \xi^* - \vec \xi)$\;
$\vec m_1^*, \vec m_2^* \gets{}$ marginals of $\vec\xi^*$\;
return $\beta, \vec m_1^*, \vec m_2^*$
}
\end{algorithm}

\section{Experiments}\label{se:experiments}

We ran experiments to evaluate our proposed algorithms on a suite of standard benchmark games, as well as two new benchmarks that we introduce. 
Each experiment was allocated $4$ threads, $64$ GB of RAM, and $6$ hours of runtime. We used Gurobi 9.5 to solve LPs and MIPs. %

\subsection{Implementation details}
The correlation DAG LP sometimes causes Gurobi's presolver to produce seemingly poorly-conditioned LPs, for reasons unknown to us. We therefore run the correlation DAG twice for each experiment, once with presolver on and once with presolver off, and report only the better of the two runtimes.

In the implementation of the two-sided column-generation algorithm (\Cref{al:colgen}), before solving a pricing problem via its MIP formulation, we try to solve the linear relaxation in which $\vec\xi\in\mc V$. If the solution to such LP is a semi-randomized correlation plan we can avoid the overhead of solving a MIP. 
Moreover, our implementation makes use of Gurobi's solution pools: since the MIP solver used for pricing problems is already tracking additional suboptimal feasible solutions, we add, together with the optimal one, such suboptimal correlation plans to $S$ with no additional computational cost. This does not affect the optimality of the final solution, and was shown to improve performances in the team games domain~\cite{Farina21:Connecting}.

\subsection{Game Instances}
We ran experiments on the following standard benchmark games. For compatibility, we use the same notation for referencing games as \citet{Zhang22:Team_DAG}. 
\begin{enumerate}
    \item \gamelbl{$^3$K$r$} is 3-player {\em Kuhn poker}~\cite{Kuhn50:Simplified} with $r$ ranks.
    \item \gamelbl{$^3$L$brs$} is 3-player {\em Leduc poker}~\cite{Southey05:Bayes} with $b$ bets per round, $r$ ranks, and $s$ suits. 
    \item \gamelbl{$^3$GL} is 3-player {\em Goofspiel}~\cite{Ross71:Goofspiel} with 3 ranks and imperfect information.
    \item \gamelbl{$^3$D$d$} is 3-player {\em Liar's Dice}~\cite{Lisy15:Online} with one $d$-sided die per player.
    \item \gamelbl{$^2$B$hwr$} is 2-player {\em Battleship}~\cite{Farina19:Correlation} on a grid of size $h \times w$, one unit-size ship per player, and $r$ rounds. 
    \item \gamelbl{$^2$S$nbr$} is a simplified version of the 2-player {\em Sheriff of Nottingham}~\cite{Farina19:Correlation} game, with $n$ items for the smuggler, a maximum bribe of $b$, and $r$ rounds of bargaining.
\end{enumerate}
Detailed rules for all of these games can be found in \citet{Farina21:Connecting} and \citet{Farina19:Correlation}. We also introduce two new parametric families of games:
\begin{enumerate}[resume]
    \item \gamelbl{$^3$T[$L$]} is a {\em trick-taking game}, which emulates the trick-taking (endgame) phase of the card game {\em bridge} where each player only has three cards remaining. When $L$ is given, $L$ deals are randomly selected at the beginning of the game, and it is common knowledge that the true deal is among them\footnote{The full game has $L = 9!/(3!)^3 = 1680$.}. \gamelbl{$^3$TP} is the {\em perfect-information} (``double-dummy'', as it is known in the bridge community) variant, which could in principle be solved by perfect-information techniques such as alpha-beta search. Nonetheless, our algorithms still run in that game, so we use them.
    
    Bridge is one of the most well-known adversarial team games. To our knowledge, computer agents in bridge have not achieved performance comparable to top humans, making it an excellent benchmark for research. The techniques in this paper obviously will not scale to the full game of bridge, but nonetheless we can show interesting results even on small endgames.%
    
    \item \gamelbl{$^2$RS$iT$} is a {\em ride-sharing game}. It is played on finite graph. Two {\em drivers} seek to earn points by reaching specific nodes of the graph and serving the requests at those nodes.  Parameter $i$ specifies the graph configuration, while $T$ is the time horizon.  
    
    Ride sharing is of course ubiquitous in the modern day. A ride-sharing company is tasked with directing its drivers in such a way that it maximizes some objective function (say, the social welfare of all drivers). But the company has no ultimate way of enforcing behavior, only recommending it. This is exactly the scenario where correlated equilibria are the right notion. Further, to our knowledge, this game is the only benchmark in the literature in general-sum correlation in which the polytope of \citet{Stengel08:Extensive} is not tight, and thus for which we know no polynomial-time algorithm. As such, it is a good testbed for our algorithms, which can run in all games. %
\end{enumerate}
Full details on our new benchmarks are given in \Cref{se:game rules}.

\subsection{Optimal Correlation}

\begin{table}[p]
      \centering
        \newcommand{\rowheight}{5.7mm}
        \def\cboxwgt{15mm}
        \def\sp{0.5mm}
      \newcommand{\cbox}[2]{%
        \definecolor{temp}{rgb}{#1}%
        \tikz[baseline,anchor=base] \node[fill=temp,text width=\cboxwgt,align=center,rounded corners=1.5pt,minimum height=\rowheight] (X) {\centering #2};%
      }
      \newcommand{\mymidrulegray}{\arrayrulecolor{lightgray}\mymidrule}
      \newcommand{\mymidruledgray}{\arrayrulecolor{gray}\mymidrule}
      \newcommand{\mymidrule}{
            \midrule
            \arrayrulecolor{black}}
      \newcommand{\unk}{\textcolor{black!30}{---}}
      \newcommand{\tworow}[1]{\multirow{2}{*}{\bf #1}}
      \setlength{\tabcolsep}{2mm}
      \scalebox{.87}{
            \begin{tabular}{@{}lrr|rrrrrrr@{}}
                  \toprule

                  \tworow{Game}                             &
                                                            &
                                                            &
                  \tworow{Concept}                          &
                  \tworow{$\abs{\mc E^c}$}                  &
                  \tworow{Value}                            &
                  \tworow{[vSF08]}                          &
                  \multicolumn{2}{c}{\bf Column generation} &
                  \centercell{\bf DAG}
                  \\[\sp]
                                                            &                &            &       &             &          &                                                                         & [FCGS21]                                                                & \llap{This paper}                                                        & \llap{This paper}
                  \\[\sp] \mymidrule
                                                            & $\abs{\mc Z}$  & 1,072      & NFCCE & 6,895       & 0.000    & \cbox{0.859361783929258,0.859361783929258,0.9718569780853518}{0.06s}    & \cbox{0.9820069204152249,0.8322952710495963,0.8322952710495963}{0.98s}  & \cbox{0.9441753171856978,0.9331795463283352,0.9331795463283352}{0.23s}   & \cbox{0.7843137254901961,0.7843137254901961,1.0}{0.02s}                 \\[\sp]
                  \gamelbl{$^2$B222}                                  & $\abs{\Sigma}$ & 11,049     & EFCCE & 20,909      & $-$0.525 & \cbox{0.8458285274894272,0.8458285274894272,0.9769319492502884}{0.15s}  & \cbox{1.0,0.7843137254901961,0.7843137254901961}{1m 57s}                & \cbox{1.0,0.7843137254901961,0.7843137254901961}{19.61s}                 & \cbox{0.7843137254901961,0.7843137254901961,1.0}{0.06s}                 \\[\sp]
                                                            & $k$            & 8          & EFCE  & 20,547      & $-$0.525 & \cbox{0.8224529027297194,0.8224529027297194,0.9856978085351787}{0.28s}  & \cbox{1.0,0.7843137254901961,0.7843137254901961}{36m 58s}               & \cbox{1.0,0.7843137254901961,0.7843137254901961}{2m 21s}                 & \cbox{0.7843137254901961,0.7843137254901961,1.0}{0.16s}                 \\[\sp]
                  \mymidrulegray
                                                            & $\abs{\mc Z}$  & 19,116     & NFCCE & 91,582      & 0.000    & \cbox{0.9589388696655132,0.8938100730488273,0.8938100730488273}{2.62s}  & \cbox{1.0,0.7843137254901961,0.7843137254901961}{13m 33s}               & \cbox{1.0,0.7843137254901961,0.7843137254901961}{11m 49s}                & \cbox{0.7843137254901961,0.7843137254901961,1.0}{0.13s}                 \\[\sp]
                  \gamelbl{$^2$B322}                                  & $\abs{\Sigma}$ & 264,541    & EFCCE & 331,310     & $-$0.317 & \cbox{0.859361783929258,0.859361783929258,0.9718569780853518}{4.20s}    & \cbox{1.0,0.7843137254901961,0.7843137254901961}{$>$ 6h}                & \cbox{1.0,0.7843137254901961,0.7843137254901961}{1h 10m}                 & \cbox{0.7843137254901961,0.7843137254901961,1.0}{1.38s}                 \\[\sp]
                                                            & $k$            & 12         & EFCE  & 503,053     & $-$0.317 & \cbox{0.8458285274894272,0.8458285274894272,0.9769319492502884}{12.97s} & \cbox{1.0,0.7843137254901961,0.7843137254901961}{$>$ 6h}                & \cbox{1.0,0.7843137254901961,0.7843137254901961}{$>$ 6h}                 & \cbox{0.7843137254901961,0.7843137254901961,1.0}{5.23s}                 \\[\sp]
                  \mymidrulegray
                                                            & $\abs{\mc Z}$  & 191,916    & NFCCE & 1,040,814   & 0.000    & \cbox{0.9677047289504037,0.8704344482891195,0.8704344482891195}{1m 5s}  & \cbox{1.0,0.7843137254901961,0.7843137254901961}{4h 26m}                & \cbox{1.0,0.7843137254901961,0.7843137254901961}{2h 44m}                 & \cbox{0.7843137254901961,0.7843137254901961,1.0}{2.35s}                 \\[\sp]
                  \gamelbl{$^2$B323}                                  & $\abs{\Sigma}$ & 3,893,341  & EFCCE & 5,321,677   & $-$0.375 & \cbox{0.8741253364090734,0.8741253364090734,0.966320645905421}{2m 13s}  & \cbox{1.0,0.7843137254901961,0.7843137254901961}{$>$ 6h}                & \cbox{1.0,0.7843137254901961,0.7843137254901961}{$>$ 6h}                 & \cbox{0.7843137254901961,0.7843137254901961,1.0}{35.48s}                \\[\sp]
                                                            & $k$            & 12         & EFCE  & 7,274,633   & $-$0.375 & \cbox{1.0,0.7843137254901961,0.7843137254901961}{oom}                   & \cbox{1.0,0.7843137254901961,0.7843137254901961}{$>$ 6h}                & \cbox{1.0,0.7843137254901961,0.7843137254901961}{oom}                    & \cbox{0.7843137254901961,0.7843137254901961,1.0}{2m 1s}                 \\[\sp]
                  \mymidrulegray
                                                            & $\abs{\mc Z}$  & 969,516    & NFCCE & 6,388,479   & 0.000    & \cbox{1.0,0.7843137254901961,0.7843137254901961}{oom}                   & \cbox{1.0,0.7843137254901961,0.7843137254901961}{oom}                   & \cbox{1.0,0.7843137254901961,0.7843137254901961}{oom}                    & \cbox{0.7843137254901961,0.7843137254901961,1.0}{20.01s}                \\[\sp]
                  \gamelbl{$^2$B324}                                  & $\abs{\Sigma}$ & 26,443,741 & EFCCE & 40,732,129  & $-$0.489 & \cbox{1.0,0.7843137254901961,0.7843137254901961}{oom}                   & \cbox{1.0,0.7843137254901961,0.7843137254901961}{oom}                   & \cbox{1.0,0.7843137254901961,0.7843137254901961}{oom}                    & \cbox{0.7843137254901961,0.7843137254901961,1.0}{10m 56s}               \\[\sp]
                                                            & $k$            & 12         & EFCE  & 49,299,490  & $-$0.489 & \cbox{1.0,0.7843137254901961,0.7843137254901961}{oom}                   & \cbox{1.0,0.7843137254901961,0.7843137254901961}{oom}                   & \cbox{1.0,0.7843137254901961,0.7843137254901961}{oom}                    & \cbox{0.7843137254901961,0.7843137254901961,1.0}{33m 12s}               \\[\sp]
                  \mymidrulegray
                                                            & $\abs{\mc Z}$  & 396        & NFCCE & 2,861       & 13.636   & \cbox{0.8790465205690119,0.8790465205690119,0.9644752018454441}{0.04s}  & \cbox{0.9907727797001153,0.8089196462898884,0.8089196462898884}{0.69s}  & \cbox{0.9672433679354094,0.8716647443291041,0.8716647443291041}{0.28s}   & \cbox{0.7843137254901961,0.7843137254901961,1.0}{0.01s}                 \\[\sp]
                  \gamelbl{$^2$S122}                                  & $\abs{\Sigma}$ & 3,717      & EFCCE & 7,385       & 9.565    & \cbox{0.8310649750096116,0.8310649750096116,0.9824682814302191}{0.04s}  & \cbox{1.0,0.7843137254901961,0.7843137254901961}{11.86s}                & \cbox{0.9815455594002307,0.8335255670895809,0.8335255670895809}{0.96s}   & \cbox{0.7843137254901961,0.7843137254901961,1.0}{0.02s}                 \\[\sp]
                                                            & $k$            & 12         & EFCE  & 6,227       & 9.078    & \cbox{0.8310649750096116,0.8310649750096116,0.9824682814302191}{0.08s}  & \cbox{1.0,0.7843137254901961,0.7843137254901961}{51.10s}                & \cbox{1.0,0.7843137254901961,0.7843137254901961}{4.09s}                  & \cbox{0.7843137254901961,0.7843137254901961,1.0}{0.04s}                 \\[\sp]
                  \mymidrulegray
                                                            & $\abs{\mc Z}$  & 2,376      & NFCCE & 17,999      & 13.636   & \cbox{0.9122645136485967,0.9122645136485967,0.9520184544405997}{0.26s}  & \cbox{1.0,0.7843137254901961,0.7843137254901961}{18.17s}                & \cbox{0.9930795847750865,0.8027681660899654,0.8027681660899654}{3.02s}   & \cbox{0.7843137254901961,0.7843137254901961,1.0}{0.04s}                 \\[\sp]
                  \gamelbl{$^2$S123}                                  & $\abs{\Sigma}$ & 33,633     & EFCCE & 69,539      & 10.000   & \cbox{0.8482891195693963,0.8482891195693963,0.9760092272202999}{0.59s}  & \cbox{1.0,0.7843137254901961,0.7843137254901961}{1h 1m}                 & \cbox{1.0,0.7843137254901961,0.7843137254901961}{3m 11s}                 & \cbox{0.7843137254901961,0.7843137254901961,1.0}{0.23s}                 \\[\sp]
                                                            & $k$            & 12         & EFCE  & 52,559      & 10.000   & \cbox{0.8421376393694733,0.8421376393694733,0.9783160322952711}{1.24s}  & \cbox{1.0,0.7843137254901961,0.7843137254901961}{48m 55s}               & \cbox{1.0,0.7843137254901961,0.7843137254901961}{7m 5s}                  & \cbox{0.7843137254901961,0.7843137254901961,1.0}{0.53s}                 \\[\sp]
                  \mymidrulegray
                                                            & $\abs{\mc Z}$  & 5,632      & NFCCE & 43,939      & 18.182   & \cbox{0.9607843137254902,0.8888888888888888,0.8888888888888888}{1.07s}  & \cbox{1.0,0.7843137254901961,0.7843137254901961}{3m 22s}                & \cbox{1.0,0.7843137254901961,0.7843137254901961}{8.82s}                  & \cbox{0.7843137254901961,0.7843137254901961,1.0}{0.05s}                 \\[\sp]
                  \gamelbl{$^2$S133}                                  & $\abs{\Sigma}$ & 95,768     & EFCCE & 165,491     & 15.000   & \cbox{0.8618223760092272,0.8618223760092272,0.9709342560553633}{2.10s}  & \cbox{1.0,0.7843137254901961,0.7843137254901961}{$>$ 6h}                & \cbox{1.0,0.7843137254901961,0.7843137254901961}{1h 28m}                 & \cbox{0.7843137254901961,0.7843137254901961,1.0}{0.67s}                 \\[\sp]
                                                            & $k$            & 12         & EFCE  & 165,859     & 15.000   & \cbox{0.8630526720492118,0.8630526720492118,0.970472895040369}{6.39s}   & \cbox{1.0,0.7843137254901961,0.7843137254901961}{$>$ 6h}                & \cbox{1.0,0.7843137254901961,0.7843137254901961}{3h 39m}                 & \cbox{0.7843137254901961,0.7843137254901961,1.0}{2.01s}                 \\[\sp]
                  \mymidrulegray
                                                            & $\abs{\mc Z}$  & 400        & NFCCE & 15,256      & 6.010    & \cbox{1.0,0.7843137254901961,0.7843137254901961}{not run}          & \cbox{0.8372164552095348,0.8372164552095348,0.980161476355248}{0.04s}   & \cbox{0.8113802383698577,0.8113802383698577,0.9898500576701269}{0.03s}   & \cbox{0.7843137254901961,0.7843137254901961,1.0}{0.02s}                 \\[\sp]
                  \gamelbl{$^2$RS12}                                  & $\abs{\Sigma}$ & 613        & EFCCE & 15,256      & 6.010    & \cbox{1.0,0.7843137254901961,0.7843137254901961}{not run}          & \cbox{0.8790465205690119,0.8790465205690119,0.9644752018454441}{0.08s}  & \cbox{0.8790465205690119,0.8790465205690119,0.9644752018454441}{0.08s}   & \cbox{0.7843137254901961,0.7843137254901961,1.0}{0.02s}                 \\[\sp]
                                                            & $k$            & 15         & EFCE  & 8,846       & 6.010    & \cbox{1.0,0.7843137254901961,0.7843137254901961}{not run}          & \cbox{0.9847750865051903,0.8249134948096886,0.8249134948096886}{0.54s}  & \cbox{0.9460207612456747,0.9282583621683967,0.9282583621683967}{0.12s}   & \cbox{0.7843137254901961,0.7843137254901961,1.0}{0.01s}                 \\[\sp]
                  \mymidrulegray
                                                            & $\abs{\mc Z}$  & 4,356      & NFCCE & 107,201,638 & 9.398    & \cbox{1.0,0.7843137254901961,0.7843137254901961}{not run}          & \cbox{0.8052287581699347,0.8052287581699347,0.9921568627450981}{3.33s}  & \cbox{0.7843137254901961,0.7843137254901961,1.0}{2.44s}                  & \cbox{0.9750865051903114,0.8507497116493656,0.8507497116493656}{1m 32s} \\[\sp]
                  \gamelbl{$^2$RS13}                                  & $\abs{\Sigma}$ & 15,063     & EFCCE & 177,846,004 & 9.385    & \cbox{1.0,0.7843137254901961,0.7843137254901961}{not run}          & \cbox{0.8347558631295655,0.8347558631295655,0.9810841983852364}{2m 11s} & \cbox{0.7843137254901961,0.7843137254901961,1.0}{1m 2s}                  & \cbox{0.8458285274894272,0.8458285274894272,0.9769319492502884}{2m 35s} \\[\sp]
                                                            & $k$            & 40         & EFCE  & 135,762,741 & 9.367    & \cbox{1.0,0.7843137254901961,0.7843137254901961}{not run}          & \cbox{1.0,0.7843137254901961,0.7843137254901961}{$>$ 6h}                & \cbox{0.9085736255286428,0.9085736255286428,0.9534025374855825}{14m 25s} & \cbox{0.7843137254901961,0.7843137254901961,1.0}{2m 18s}                \\[\sp]
                  \mymidrulegray
                                                            & $\abs{\mc Z}$  & 484        & NFCCE & 53,983      & 7.188    & \cbox{1.0,0.7843137254901961,0.7843137254901961}{not run}          & \cbox{0.8003075740099962,0.8003075740099962,0.994002306805075}{0.08s}   & \cbox{0.8901191849288735,0.8901191849288735,0.960322952710496}{0.29s}    & \cbox{0.7843137254901961,0.7843137254901961,1.0}{0.06s}                 \\[\sp]
                  \gamelbl{$^2$RS22}                                  & $\abs{\Sigma}$ & 701        & EFCCE & 53,983      & 7.176    & \cbox{1.0,0.7843137254901961,0.7843137254901961}{not run}          & \cbox{0.8347558631295655,0.8347558631295655,0.9810841983852364}{0.13s}  & \cbox{0.8175317185697808,0.8175317185697808,0.9875432525951557}{0.10s}   & \cbox{0.7843137254901961,0.7843137254901961,1.0}{0.06s}                 \\[\sp]
                                                            & $k$            & 15         & EFCE  & 31,503      & 7.176    & \cbox{1.0,0.7843137254901961,0.7843137254901961}{not run}          & \cbox{0.9441753171856978,0.9331795463283352,0.9331795463283352}{0.56s}  & \cbox{0.8901191849288735,0.8901191849288735,0.960322952710496}{0.24s}    & \cbox{0.7843137254901961,0.7843137254901961,1.0}{0.05s}                 \\[\sp]
                  \mymidrulegray
                                                            & $\abs{\mc Z}$  & 4,096      & NFCCE & oom         & 10.961   & \cbox{1.0,0.7843137254901961,0.7843137254901961}{not run}          & \cbox{0.7843137254901961,0.7843137254901961,1.0}{2.65s}                 & \cbox{0.7990772779700115,0.7990772779700115,0.9944636678200692}{3.33s}   & \cbox{1.0,0.7843137254901961,0.7843137254901961}{oom}                   \\[\sp]
                  \gamelbl{$^2$RS23}                                  & $\abs{\Sigma}$ & 13,277     & EFCCE & oom         & 10.820   & \cbox{1.0,0.7843137254901961,0.7843137254901961}{not run}          & \cbox{0.8507497116493656,0.8507497116493656,0.9750865051903114}{1m 51s} & \cbox{0.7843137254901961,0.7843137254901961,1.0}{42.03s}                 & \cbox{1.0,0.7843137254901961,0.7843137254901961}{oom}                   \\[\sp]
                                                            & $k$            & 44         & EFCE  & oom         & 10.791   & \cbox{1.0,0.7843137254901961,0.7843137254901961}{not run}          & \cbox{1.0,0.7843137254901961,0.7843137254901961}{$>$ 6h}                & \cbox{0.7843137254901961,0.7843137254901961,1.0}{7m 11s}                 & \cbox{1.0,0.7843137254901961,0.7843137254901961}{oom}                   \\[\sp]
                  \bottomrule
            \end{tabular}
      }
      \caption{Experiments on general-sum correlated equilibria, comparing both our correlation DAG LP and two-sided column generation to earlier approaches. {\bf vSF08} is the relaxation of \citet{Stengel08:Extensive}, which is only correct in triangle-free games~\cite{Farina20:Polynomial}. RS is not triangle-free, so vSF fails in that game. {\bf FCGS21} is the one-sided column generation approach of \citet{Farina21:Connecting}. $k$ is the information complexity. All runs were performed to convergence.  `oom' means out of memory. Runtimes are colored according to the ratio with the best runtime in that row, according to the scale \includegraphics[scale=.65,valign=t, trim=8 0 6 4, clip]{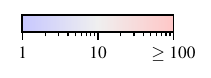}.}
      \label{ta:correl 2p body}
\end{table}

\begin{table}[t]
\centering
        \newcommand{\rowheight}{5.7mm}
        \def\cboxwgt{15mm}
        \def\sp{0.9mm}
      \newcommand{\cbox}[2]{%
        \definecolor{temp}{rgb}{#1}%
        \tikz[baseline,anchor=base] \node[fill=temp,text width=\cboxwgt,align=center,rounded corners=1.5pt,minimum height=\rowheight] (X) {\centering #2};%
      }
\newcommand{\mymidrulegray}{\arrayrulecolor{gray}\mymidrule}
\newcommand{\mymidrule}{
\midrule
\arrayrulecolor{black}}
\newcommand{\unk}{\textcolor{black!30}{---}}
\setlength{\tabcolsep}{2mm}
\colorlet{teamcol}{gray}
\newcommand{\teamprint}[1]{\{\textcolor{teamcol}{#1}\}}
\newcommand{\oom}{\cbox{1.0,0.7843137254901961,0.7843137254901961}{oom}}
\newcommand{\tworow}[1]{\multirow{2}{*}{\bf #1}}
\setlength{\tabcolsep}{4mm}
\scalebox{.85}{
\begin{tabular}{lrrrrrrrrrrrrr}
\toprule
    
      \tworow{Game} & \tworow{$\{\ptwo\}$} &
      \tworow{Leaves} &
      \tworow{\pone Value} &
      \tworow{$k$} &
      \bf DAG CFR  &
      \multicolumn{2}{c}{\bf Column generation}
      \\[\sp]
      &&&&& [ZFS22] & [FCGS21] & This paper
      \\[\sp] \mymidrule
    \gamelbl{$^3$K3} & \teamprint{3} &      78 &    0.000 &  6 &   \cbox{0.9414071510957324,0.940561322568243,0.940561322568243}{0.00s} &     \cbox{0.9414071510957324,0.940561322568243,0.940561322568243}{0.00s} &   \cbox{0.9414071510957324,0.940561322568243,0.940561322568243}{0.00s} \\[\sp]
    \gamelbl{$^3$K4} & \teamprint{3} &     312 & $-$0.042 &  8 &   \cbox{0.9414071510957324,0.940561322568243,0.940561322568243}{0.00s} &     \cbox{0.9547866205305652,0.904882737408689,0.904882737408689}{0.02s} &   \cbox{0.9414071510957324,0.940561322568243,0.940561322568243}{0.01s} \\[\sp]
    \gamelbl{$^3$K6} & \teamprint{3} &   1,560 & $-$0.024 & 12 & \cbox{0.8445982314494425,0.8445982314494425,0.9773933102652825}{0.01s} &   \cbox{0.9746251441753172,0.8519800076893502,0.8519800076893502}{0.15s} &   \cbox{0.9414071510957324,0.940561322568243,0.940561322568243}{0.04s} \\[\sp]
    \gamelbl{$^3$K8} & \teamprint{3} &   4,368 & $-$0.019 & 16 &                \cbox{1.0,0.7843137254901961,0.7843137254901961}{0.79s} &     \cbox{0.9852364475201846,0.823683198769704,0.823683198769704}{0.36s} &   \cbox{0.9414071510957324,0.940561322568243,0.940561322568243}{0.06s} \\[\sp]
   \gamelbl{$^3$K12} & \teamprint{3} &  17,160 & $-$0.014 &  24 &                                                                   \oom &   \cbox{0.9681660899653979,0.8692041522491349,0.8692041522491349}{1.24s} &   \cbox{0.9414071510957324,0.940561322568243,0.940561322568243}{0.43s} \\ \mymidrulegray
  \gamelbl{$^3$L132} & \teamprint{3} &   4,500 &    0.293 &  6 &                \cbox{0.7843137254901961,0.7843137254901961,1.0}{0.01s} &  \cbox{0.9981545559400231,0.7892349096501345,0.7892349096501345}{52.66s} &   \cbox{0.9414071510957324,0.940561322568243,0.940561322568243}{5.75s} \\[\sp]
  \gamelbl{$^3$L133} & \teamprint{3} &   6,477 &    0.215 &  6 &                \cbox{0.7843137254901961,0.7843137254901961,1.0}{0.01s} &                 \cbox{1.0,0.7843137254901961,0.7843137254901961}{1m 16s} &   \cbox{0.9414071510957324,0.940561322568243,0.940561322568243}{7.71s} \\[\sp]
  \gamelbl{$^3$L151} & \teamprint{3} &  10,020 & $-$0.019 & 10 &                \cbox{0.7843137254901961,0.7843137254901961,1.0}{0.07s} &  \cbox{0.9898500576701269,0.8113802383698577,0.8113802383698577}{22.03s} &   \cbox{0.9414071510957324,0.940561322568243,0.940561322568243}{3.29s} \\[\sp]
  \gamelbl{$^3$L223} & \teamprint{3} &   8,762 &    0.516 &  4 &                \cbox{0.7843137254901961,0.7843137254901961,1.0}{0.01s} &                 \cbox{1.0,0.7843137254901961,0.7843137254901961}{2m 13s} &   \cbox{0.9414071510957324,0.940561322568243,0.940561322568243}{4.61s} \\[\sp]
  \gamelbl{$^3$L523} & \teamprint{3} & 775,148 &    0.953 &  4 &                \cbox{0.7843137254901961,0.7843137254901961,1.0}{3.60s} &                 \cbox{1.0,0.7843137254901961,0.7843137254901961}{$>$ 6h} &  \cbox{0.9414071510957324,0.940561322568243,0.940561322568243}{4h 39m} \\ \mymidrulegray
    \gamelbl{$^3$D2} & \teamprint{3} &     504 &    0.250 &  4 & \cbox{0.8310649750096116,0.8310649750096116,0.9824682814302191}{0.00s} &    \cbox{0.960322952710496,0.8901191849288735,0.8901191849288735}{0.11s} &   \cbox{0.9414071510957324,0.940561322568243,0.940561322568243}{0.05s} \\[\sp]
    \gamelbl{$^3$D3} & \teamprint{3} &  13,797 &    0.284 &  6 &                \cbox{0.7843137254901961,0.7843137254901961,1.0}{0.06s} &                 \cbox{1.0,0.7843137254901961,0.7843137254901961}{5m 13s} &   \cbox{0.9414071510957324,0.940561322568243,0.940561322568243}{5.95s} \\ \mymidrulegray
    \gamelbl{$^3$GL} & \teamprint{3} &   1,296 &    1.252 &  2 &                \cbox{0.7843137254901961,0.7843137254901961,1.0}{0.00s} &   \cbox{0.9741637831603229,0.8532103037293348,0.8532103037293348}{0.78s} &   \cbox{0.9414071510957324,0.940561322568243,0.940561322568243}{0.21s} \\ \mymidrulegray
 \gamelbl{$^3$T[50]} & \teamprint{2} &  10,300 &    0.600 &  5 &                \cbox{0.7843137254901961,0.7843137254901961,1.0}{0.00s} &                 \cbox{0.984313725490196,0.8261437908496732,0.8261437908496732}{5.49s} &   \cbox{0.9414071510957324,0.940561322568243,0.940561322568243}{0.67s} \\[\sp]
\gamelbl{$^3$T[100]} & \teamprint{2} &  20,992 &    0.710 & 18 &                \cbox{0.7843137254901961,0.7843137254901961,1.0}{0.07s} &  \cbox{0.9893886966551326,0.8126105344098423,0.8126105344098423}{10.54s} &   \cbox{0.9414071510957324,0.940561322568243,0.940561322568243}{1.62s} \\[\sp]
\gamelbl{$^3$T[840]} & \teamprint{2} & 190,228 &    0.569 &  141 &                                                                   \oom & \cbox{0.9972318339100346,0.7916955017301038,0.7916955017301038}{16m 40s} &  \cbox{0.9414071510957324,0.940561322568243,0.940561322568243}{1m 51s} \\[\sp]
     \gamelbl{$^3$T} & \teamprint{2} & 379,008 &   $\approx$ 0.573 &  141 &                                                                   \oom &                 \cbox{1.0,0.7843137254901961,0.7843137254901961}{4h 10m} & \cbox{0.9414071510957324,0.940561322568243,0.940561322568243}{22m 58s} \\[\sp]
    \gamelbl{$^3$TP} & \teamprint{2} & 379,008 &    0.658 &  2 &                \cbox{0.7843137254901961,0.7843137254901961,1.0}{0.49s} &  \cbox{0.9893886966551326,0.8126105344098423,0.8126105344098423}{6m 55s} &   \cbox{0.9414071510957324,0.940561322568243,0.940561322568243}{1m 3s} \\[\sp]
\bottomrule
\end{tabular}
}
\caption{Experiments on TMECor in adversarial team games, comparing our two-sided column generation approach to earlier approaches. {\bf DAG CFR} is the CFR-based team DAG algorithm of  \citet{Zhang22:Team_DAG}. {\bf FCGS21} is the one-sided column generation approach of \citet{Farina21:Connecting}. All runtimes are reported to a target precision of $0.005$ times the reward range of the game. The game value of $^3$T is after our new incremental algorithm ran to the time limit, and is accurate to $\pm 0.002$. All other game values are accurate to three decimals. Runtimes are colored according to the ratio with the runtime of our two-sided column generation, according to the scale \includegraphics[scale=.65,valign=t, trim=8 0 6 4, clip]{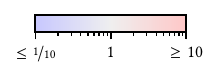}.}
\label{ta:tmecor experiments}
\end{table}

\begin{figure}[tbp]\centering
\payoffspace B2{323}
\payoffspace S2{133}
\payoffspace{RS}{2}{12}
\\
\payoffspace D3{2}
\payoffspace{GL}{3}{}
\payoffspace K35
\\
\payoffspace L3{223}
\payoffspace T3{50}
\payoffspace{TP}{3}{}
\\
\makelegend

\caption{Payoff spaces in several games, with all three notions of equilibrium. More such plots can be found in the appendix. }\label{fi:payoff spaces}
\end{figure}

We evaluated the performance of the DAG-based LP and the two-sided column-generation framework against the prior state-of-the art algorithms for computing optimal correlated equilibria in general-sum extensive-form games: the polytope of~\citet{Stengel08:Extensive} (denoted by {\bf [vSF08]}), which is correct only for a certain family of games called {\em triangle-free games} (we denote with `n/a' when this is not the case), and the one-sided column-generation algorithm by \citet{Farina21:Connecting} (denoted by {\bf [FCGS21]}), which we adapted from the team domain.

Table~\ref{ta:correl 2p body}  summarizes the comparison over two-player game instances. As expected, the correlation DAG LP has the best running times for games with small information complexity parameter $k$. When this is the case, it dramatically outperforms previous algorithms: it can solve in a matter of seconds instances that previously exceeded 6 hours (see, \textit{e.g.}, \gamelbl{$^2$B323} and \gamelbl{$^2$S133}), and it can solve in less than 1 hour instances that previously were not computationally feasible (\textit{e.g.}, \gamelbl{$^2$B324}). On the other hand, when $k$ is large (\textit{e.g.}, in \gamelbl{$^2$RS23}), the two-sided column-generation algorithm provides the best running times. For example, when computing optimal EFCE, it requires 6 minutes while the prior one-sided column-generation algorithm {\bf [vSF08]} exceeds 6 hours. Combining the two techniques that we propose yields uniformly better performance than  prior work for any value of the parameter $k$.

We also ran experiments on three-player games using the correlation DAG LP (see Table~\ref{ta:correl 3p body} in the appendix). This shows, for the first time, that it is possible to compute optimal NFCCE/EFCCE/EFCE in practice for large game instances even when the number of players is greater than two.

\subsection{Payoff Space Plots}

In \Cref{fi:payoff spaces}, we show plots of the space of feasible payoffs in several tested games.  All three-player games we tested on were constant sum, so for those games we show a 2D projection. 

In most games tested, all three payoff spaces are different, and show very detailed boundaries that almost seem smooth (though, of course, they cannot be, since the payoff space is a polytope). This confirms the findings of earlier papers, \textit{e.g.},~\citet{Farina20:Coarse}, and demonstrates the importance of defining the various notions as separate.

\subsection{Adversarial Team Games}
We compared our new column generation approach (\Cref{al:colgen}) to prior approaches for finding team-correlated equilibria (TMECor) in zero-sum adversarial team games. Specifically, we compared to the CFR algorithm on the team DAG introduced by \citet{Zhang22:Team_DAG} and the prior column generation-based approach of \citet{Farina21:Connecting}. Results can be found in \Cref{ta:tmecor experiments}. Our results clearly give several conclusions. First, our algorithm is an improvement upon \citet{Farina21:Connecting}, achieving speedups of more than an order of magnitude in some games. Second, our algorithm, like \citet{Zhang21:Computing}, scales well in the information complexity $k$ compared to that of \citet{Zhang22:Team_DAG}: while ours is slower when $k$ is small, it begins to match and quickly exceed the performance of that algorithm when $k$ grows larger, as happens in Kuhn poker. 

In the Tricks game instances, we observe that the perfect-information value, $0.66$, does not match the team game value, $0.57$. The discrepancy of nearly $0.1$ tricks is nontrivially large given that there are only three tricks remaining. This establishes that, even in small endgames with three cards left, the fact that players do not know the cards of their teammate or opponent is still relevant information in a game of bridge, showing the importance of viewing bridge as a true imperfect-information game between two teams, rather than as a perfect-information game as {\em double dummy} bridge endgame solvers do~[\textit{e.g.},~\citenum{ginsberg1999gib}]. 

\section{Conclusions and Future Research}

In this paper, we introduced and analyzed two new approaches for finding optimal correlated equilibria in general-sum games: the {\em correlation DAG} and a {\em two-sided column generation} algorithm. The former has complexity parameterizable by the {\em information complexity} $k$ of the game. The two techniques have complementary practical strengths and weaknesses: when $k$ is small, the correlation DAG shines; when $k$ grows large, the column generation technique is faster and more frugal in terms of memory usage. Furthermore, the value of $k$ can be easily computed, enabling an efficient choice between the two approaches. Our techniques are the state of the art in practice across all the games tested (except two small games where the prior column-generation approach was slightly faster). Possible directions of future research include the following.
\begin{enumerate}
    \item Extend the correlation DAG in such a way that it also has polynomial size in all triangle-free games.
    \item An intelligent combination---rather than merely a selection of one versus the other---of the correlation DAG and the column generation algorithm may lead to faster practical algorithms.
    \item Investigate possible use of the payoff structure in the game; for example, investigate extensions of the concept of {\em smooth games}~\cite{Roughgarden15:Intrinsic}.
    \item Our algorithms for optimal correlation all ultimately reduce to linear programs or mixed-integer programs. On the other hand, as we have discussed, regret minimization algorithms are known to be able to find {\em one} correlated equilibrium in all the notions we discuss in the paper, as well as equilibria in adversarial team games. We leave it to future research to answer whether regret minimization---with exact utility gradient or even with sampling (when only a gradient estimate is available)---can be made to lead to \textit{optimal} correlated equilibria. 
\end{enumerate}

\revision{Since the publication of the initial version of this paper, some developments have been made along these lines. \citet{Zhang24:DAG} developed a method for combining the team belief DAG with column generation, which can also be applied to optimal equilibria, leading to ``best-of-both-worlds'' performance that combines the strengths of the two approaches. \citet{Peng24:Fast} and \citet{Dagan24:External} have shown that, for the stronger notion of NFCE and absolute constant gap $\eps$, there exists a learning-based polynomial-time algorithm for computing an $\eps$-NFCE.  (It remains an open question whether there exists such an algorithm for $\eps < 1/\poly(|\mc H|)$.) \citet{Zhang23:Computing} developed learning-based algorithms for computing optimal equilibria including---but not limited to---NFCCE, EFCCE, and EFCE that essentially match the parameterized complexity guarantees offered by the present paper by taking a Lagrangian relaxation of \eqref{eq:orig} and viewing it as a zero-sum game. }

\section*{Acknowledgements}
We thank the anonymous reviewers for their valuable suggestions, which greatly improved the manuscript. We also thank Bernhard von Stengel for stimulating discussions that greatly improved the presentation of the material. This material is based on work supported by the Vannevar Bush Faculty
Fellowship ONR N00014-23-1-2876; the National Science Foundation under grants IIS-1901403, CCF-1733556, RI-2312342, and RI-1901403; the ARO under awards W911NF2010081 and W911NF2210266; and NIH award
A240108S001.

\bibliographystyle{plainnat}
\bibliography{dairefs}

\appendix

\section{Team Belief DAG}\label{app:tb-dag}

In the interest of self-containment, in this section we give a description of the construction of \citet{Zhang22:Team_DAG}. Given a timeable extensive-form game $\Gamma$ and a player $i$ with imperfect recall, the goal is to represent the sequence form of player $i$. To do this, we construct a DAG $\mc D$ with two types of nodes: {\em decision nodes}, at which $i$ makes a decision, and {\em observation nodes}, at which $i$ observes something. Nodes in $\mc D$ will be identified with {\em sets} of histories in $\Gamma$, with the exception that we allow an observation node and a decision node to have the same set.
\begin{itemize}
    \item The root node of $\mc D$ is the decision node $\{ \Root \}$, where $\Root$ is the root node of $\Gamma$.
    \item Decision nodes $\{z \}$ for $z \in Z$ are leaves of $\mc D$.
    \item At a decision node $B$, let $I_1, \dots, I_m \in \mc I_i$ be the infosets with nonempty intersection with $B$. Player $i$ may pick any {\em prescription} $$\vec a \in \bigtimes_{i \in [m]} A_{I_i}.$$ The next node is the observation node
    \begin{align}
        B \vec a := \{ h \vec a_i : h \in I_i \cap B \} \cup \{ ha : h \in B \setminus \mc H_i, a \in A_h \}
    \end{align}
    \item At an observation node $O$, let $B_1, \dots, B_m$ be the connected components of the induced subgraph $G_i[O]$, where $G_i$ is player $i$'s connectivity graph (\Cref{def:connectivity}). 
\end{itemize}
The critical observation of \citet{Zhang22:Team_DAG} (which we will not explicitly prove here) is that the set of sequence-form mixed strategies for player $i$ in $\Gamma$ is a projection of the set of strategies of the player in the decision problem $\mc D$, and moreover the latter set has a representation using linear constraints with size $O^*(|\mc D|)$, where $|\mc D|$ denotes the number of nodes in $\mc D$. \Cref{th:dag} then follows by observing that the decision nodes in $\mc D$ are precisely the beliefs $B \in \mc B_i$ and for each such belief $B$ there are $\prod_{I \in \mc I_i : B \cap \mc I_i \ne \emptyset} A_I$ actions.
\section{Omitted Proofs}
\subsection{\Cref{th:public actions} (Games with Public Actions)}\label{se:pr:public actions}
\thmPublicActions*
\begin{proof}
  \citet{Zhang22:Team_DAG} devise an algorithm for constructing, starting with a game $\Gamma$ with public actions, a new strategically equivalent game $\Gamma'$ with branching factor $2$, no higher parameter $k$, and at most polynomially larger. The method works by breaking up each high-branching-factor node into several successive binary decisions, in such a way that public state size is preserved. For NFCCE, this is sufficient to immediately conclude the desired result. For EFCCE, it suffices to additionally observe that one only needs to care about trigger histories $h^\tau$ in $\Gamma'$ where $\tau$ is a valid trigger in $\Gamma$. The number of these is at most the depth $d$ of the original game $\Gamma$.
\end{proof}

\section{Games in Experiments}\label{se:game rules}

\subsection{Trick-Taking Game (Bridge Endgame)}

We introduce a {\em trick-taking game}, which is effectively a bridge endgame scenario. There is a fixed deck of playing cards consisting of 3 ranks ($2, 3, 4$) of each of four suits (\cS{}, \cH{}, \cD{}, \cC{}). Spades (\cS{}) is designated as the {\em trump suit}. There are four players: two {\em defenders}, who sit across from each other at the table, the {\em dummy}, and the {\em declarer}. The actions of the dummy will be controlled by the declarer; as such, there are actually only three players in the game. However, in this section, we will use the four-player terminology because it is easier to understand.

The whole deck is randomly dealt to four players. The dummy's cards are then publicly revealed. Play proceeds in {\em tricks}. The player to the left of the declarer leads the first trick. In each trick, the leader of the trick first plays a card. The suit of that card is the {\em lead suit}. Then, in clockwise order around the table, the other three players play a card from their hand. Players must play a card of the lead suit if they have such a card; otherwise, they may play any card. If any \cS{} has been played, then whoever plays the highest \cS{} wins the trick. Otherwise, the highest card of the lead suit wins the trick. The winner of one trick leads the next trick. At the end of the game, each player earns as many points as tricks they have won. For the adversarial team game, the two defenders are teammates, playing against the declarer (who controls the dummy).

We use \gamelbl{$^3$T} to refer to the trick-taking game. In the {\em perfect information} variant \gamelbl{$^3$TP}, all information is public, creating a perfect-information game. This is equivalent to what the bridge community calls a {\em double dummy} game. In all our games, the dummy's hand is fixed as \cS{2} \cH{2} \cH{3}.

In the {\em limited deals} variant \gamelbl{$^3$T[$L$]}, $L$ deals are randomly selected at the beginning of the game,  and it is common knowledge that the true deal is among them. This limits the size of the game tree, as well as the parameters on which the complexity of our algorithms depend. $L = 9!/(3!)^3 = 1680$ is the full game.

\subsection{Ride-Sharing Game Instances}

We introduce a new benchmark which we call \emph{ride-sharing game}.

\paragraph{General rules of the game} The game models the interaction between two players (\emph{a.k.a.}, drivers), which compete to serve requests on a road network. In particular, the network is modeled as an undirected graph $\Grs=(\Vrs,\Ers)$. Each vertex $v\in \Vrs$ corresponds to a ride request to be served. Each ride request has a reward in $\mathbb{R}_{\ge 0}$. Each edge in the road network has some cost (representing the time incurred to traverse the edge).
The first driver who arrives on node $v\in \Vrs$ serves the corresponding ride, and receives the corresponding reward. Once a node has been served, it stays clean until the end of the game. The game terminates when all requests have been served, or when a timeout is met (\textit{i.e.}, there's a fixed time horizon $T$). If the two drivers arrive on the same vertex at the same time they get reward 0.
The final utility of each driver is the sum of the rewards obtained from the beginning until the end of the game.
The initial position of the two drivers is randomly selected at the beginning of the game. Finally, the two drivers can observe each other's position only when they are simultaneously on the same node, or they are in adjacent nodes.

\paragraph{Objective and remarks} Ride-sharing games are particularly well-suited to study the computation of optimal correlated equilibria because they are two-player, general-sum games which are \emph{not} triangle-free~\cite{Farina20:Polynomial}. That is not the case for some of the existing two-player general-sum benchmarks, such as Goofspiel.
We take the perspective of a \emph{centralized platform} that has the goal of steering the drivers' behavior so as to maximize the overall social welfare. The platform can send \emph{recommendations} to players in the form of navigation instructions. The goal of the platform is to ensure that such recommendations are incentive compatible, and maximize the SW attained at the equilibrium. Depending on the type of interaction in place between the platform and the players, the platform's goal amounts to finding an optimal (\textit{i.e.}, social-welfare maximizing) NFCCE/EFCCE/EFCE. For example, if the platform implemented an EFCE-like interaction protocol, at each new vertex in $\Vrs$ a driver would receive a suggestion about the next road to take from there. The driver would be free to deviate as such decision point, since they could decide to take another direction, and that would come at the cost of future recommendations.

\paragraph{Implementation details} In our experiments, we employ road networks with unitary cost associated to edges. We write \gamelbl{$^n$RS$iT$} to indicate the game instance has $n$ players, and was generated from map $i$ with time horizon $T$ (\textit{i.e.}, each driver can make at most $T$ steps). We employ two maps (map 1 and map 2), and we generate the instances \gamelbl{$^2$RS13}, \gamelbl{$^2$RS14}, \gamelbl{$^2$RS23}.
In Figure~\ref{fig:rs} we report the structure of the two maps. The value between curly brackets is the reward for a request on that node.

\begin{figure}[p]
  \centering
  \scalebox{.7}{
    \begin{tikzpicture}[shorten >=1pt,auto,node distance=3cm,
        thick,main node/.style={circle,draw,
            font=\sffamily\Large\bfseries,minimum size=5mm}]

      \node[main node, label={[align=left]\{1\}}] at (0, 0)   (0) {0};
      \node[main node, label={[align=left]\{0.5\}}] at (1.5, 2)   (1) {1};
      \node[main node, label={[left,label distance=.15cm]\{.5\}}] at (1.5, 0)   (2) {3};
      \node[main node, label={[anchor=north]below:\{1.5\}}] at (1.5, -2)   (3) {2};
      \node[main node, label={[align=left]\{4.5\}}] at (3, 0)   (4) {4};
      \node[main node, label={[align=left]\{2\}}] at (4, 2)   (5) {5};
      \node[main node, label={[anchor=north]below:\{1.5\}}] at (4, -2)   (6) {6};

      \path[every node/.style={font=\sffamily\small,
            fill=white,inner sep=1pt}]

      (0) edge[-] (1)
      (0) edge[-] (3)
      (1) edge[-] (2)
      (3) edge[-] (2)
      (4) edge[-] (2)
      (3) edge[-] (6)
      (5) edge[-] (6)
      (5) edge[-] (1)
      (2) edge[-] (5)
      (2) edge[-] (6);
    \end{tikzpicture}
  }
  \hspace{1cm}
  \scalebox{.7}{
    \begin{tikzpicture}[shorten >=1pt,auto,node distance=3cm,
        thick,main node/.style={circle,draw,
            font=\sffamily\large\bfseries,minimum size=5mm}]

      \node[main node, label={[align=left]\{1\}}] at (0, 0)   (0) {0};
      \node[main node, label={[align=left]\{0.5\}}] at (2, 3)   (1) {1};
      \node[main node, label={[align=left]\{0.5\}}] at (2, 1.2)   (2) {2};
      \node[main node, label={[align=left]right:\{1.5\}}] at (2, -1.5)   (3) {3};
      \node[main node, label={[align=left]\{1\}}] at (2, -3)   (4) {4};
      \node[main node, label={[align=right]right:\{2.5\}}] at (3.5, 2.25)   (5) {5};
      \node[main node, label={[align=right]right:\{1.5\}}] at (3.5, -2.25)   (6) {6};
      \node[main node, label={[align=right]right:\{5\}}] at  (3.5, 0)  (7) {7};

      \path[every node/.style={font=\sffamily\small,
            fill=white,inner sep=1pt}]

      (0) edge[-] (1)
      (0) edge[-] (2)
      (0) edge[-] (3)
      (0) edge[-] (4)
      (5) edge[-] (1)
      (5) edge[-] (2)
      (0) edge[-] (1)
      (3) edge[-] (2)
      (3) edge[-] (6)
      (4) edge[-] (6)
      (0) edge[-] (1)
      (5) edge[-] (7)
      (7) edge[-] (6);
    \end{tikzpicture}

  }
  \caption{The two road network configurations which we consider. \emph{Left}: map 1 (used for \gamelbl{$^2$RS13}, \gamelbl{$^2$RS14}). \emph{Right}: map 2 (used for \gamelbl{$^2$RS23}). In both cases the position of the two drivers is randomly chosen at the beginning of the game, edge costs are unitary, and one reward per node is indicated between curly brackets.}
  \label{fig:rs}
\end{figure}

\begin{table}[p]
\centering
\newcommand{\cbox}[2]{\fcolorbox{white}[rgb]{#1}{\phantom{00m 00s}\llap{#2}}}
\newcommand{\mymidrulegray}{\arrayrulecolor{lightgray}\mymidrule}
\newcommand{\mymidruledgray}{\arrayrulecolor{gray}\mymidrule}
\newcommand{\mymidrule}{
\midrule
\arrayrulecolor{black}}
\newcommand{\unk}{\textcolor{black!30}{---}}
\newcommand{\tworow}[1]{\multirow{2}{*}{\bf #1}}

\scalebox{0.8}{
\begin{tabular}{lrr|rrrrrrr}
\toprule
      \tworow{Game} & \tworow{Leaves} & \tworow{$k$} & \tworow{Concept} & \bf \tworow{$\abs{\mc E^c}$} & \tworow{Runtime} & \multicolumn{3}{c}{\bf Optimal value} \\
      &&&&&& \centercell{P1} & \centercell{P2} & \centercell{P3}
      \\ \mymidrule
          &         &    & NFCCE &     259,176 &  0.63s & $-$0.018 & $-$0.007 & 0.064 \\
\gamelbl{$^3$K4} &     312 & 12 & EFCCE &     370,408 &  0.87s & $-$0.020 & $-$0.012 & 0.057 \\
          &         &    &  EFCE &     249,508 &  0.72s & $-$0.021 & $-$0.013 & 0.055 \\
\mymidrulegray
          &         &    & NFCCE &   4,182,981 & 18.44s & $-$0.011 &    0.017 & 0.057 \\
\gamelbl{$^3$K5} &     780 & 15 & EFCCE &   7,236,161 & 49.89s & $-$0.016 &    0.015 & 0.052 \\
          &         &    &  EFCE &   5,150,241 & 39.14s & $-$0.016 &    0.013 & 0.052 \\
\mymidrulegray
          &         &    & NFCCE &     605,941 &  2.32s &    2.450 &    2.197 & 2.072 \\
 \gamelbl{$^3$L223} &   8,762 &  6 & EFCCE &   6,342,970 & 34.84s &    1.302 &    1.431 & 1.309 \\
          &         &    &  EFCE &   5,251,772 & 50.56s &    0.877 &    1.009 & 1.000 \\
\mymidrulegray
          &         &    & NFCCE &      34,212 &  0.05s &    0.250 &    0.250 & 0.131 \\
   \gamelbl{$^3$D2} &     504 &  6 & EFCCE &      54,627 &  0.08s &    0.250 &    0.250 & 0.000 \\
          &         &    &  EFCE &     621,237 &  1.57s &    0.250 &    0.250 & 0.000 \\
\mymidrulegray
          &         &    & NFCCE &      29,865 &  0.04s &    2.505 &    2.505 & 2.505 \\
   \gamelbl{$^3$GL} &   1,296 & 10 & EFCCE &      36,933 &  0.06s &    2.476 &    2.476 & 2.476 \\
          &         &    &  EFCE &      16,950 &  0.06s &    2.467 &    2.467 & 2.467 \\
\mymidrulegray
          &         &    & NFCCE & 154,973,683 & 1m 44s &    1.463 &    1.380 & 0.887 \\
\gamelbl{$^3$T[50]} &  10,300 & 15 & EFCCE & 155,188,423 & 1m 40s &    1.420 &    1.360 & 0.840 \\
          &         &    &  EFCE & 155,340,357 & 1m 31s &    1.420 &    1.360 & 0.840 \\
\mymidrulegray
          &         &    & NFCCE &   6,678,856 & 12.91s &    1.466 &    1.477 & 1.037 \\
   \gamelbl{$^3$TP} & 379,008 &  3 & EFCCE &  11,625,688 & 28.73s &    1.451 &    1.442 & 0.922 \\
          &         &    &  EFCE &   6,714,256 & 17.13s &    1.451 &    1.442 & 0.922 \\
\bottomrule
\end{tabular}
}
\caption{Experiments on general-sum correlated equilibria in 3-player games, with the correlation DAG. All runs were performed to convergence. Since all games tested are constant sum, instead of reporting the social welfare optimum (which is always the constant sum), we report the optimal utility for each individual agent in every solution concept.}
\label{ta:correl 3p body}
\end{table}

\noindent
\begin{figure}[p]\centering
  \payoffspace B2{222}
  \payoffspace B2{322}
  \payoffspace B2{323}
  \\
  \payoffspace B2{324}
  \payoffspace S2{122}
  \payoffspace S2{123}
  \\
  \payoffspace S2{133}
  \payoffspace{RS}{2}{12}
  \payoffspace{RS}{2}{22}
  \\
  \makelegend
  \caption{Payoff spaces in two-player games.}
\end{figure}

\begin{figure}[p]\centering
  \payoffspace D3{2}
  \payoffspace{GL}{3}{}
  \payoffspace K3{4}
  \\
  \payoffspace K35
  \payoffspace L3{223}
  \payoffspace T3{50}
  \\
  \payoffspace{TP}{3}{}
  \\
  \makelegend
  \caption{Payoff spaces in three-player, fixed-sum games.}
\end{figure}

\begin{table}[t]
\newcommand{\header}[1]{\multicolumn{2}{l}{\textbf{{#1}}}\\\midrule}
    \begin{tabular}{ll}
    \header{Games}
       $[n] = \{ 1, ..., n \}$ & the number of players
    \\ $h \in \mc H$ & a node in a game
    \\ $\Root$ & the root node
    \\ $z \in \mc Z$ & a terminal node
    \\ $A_h$, $A_I$ & the set of actions available at node $h$ or infoset $I$
    \\ $a$ & an action
    \\ $ha$ & the child of $h$ reached by playing $a$
    \\ $\mc H_i$ & the set of nodes at which player $i$ (possibly chance) acts
    \\ $p(z)$ & the probability that chance plays all actions on the path to $z$
    \\ $\sigma_i \in \Sigma_i$ & a sequence of player $i$
    \\ $\sigma_i(h)$ & the sequence of player $i$ at node $h$
    \\ $\vec x_i[s]$ & the probability that player $i$ plays to the given sequence or node $s$
    \\ $\vec x_i \in \mc X_i$ & a (sequence-form) mixed strategy of player $i$
    \\ $u_i(z), u_i(\vec x), \cdots$ & the (expected) utility of player $i$
    \\ $\preceq$ & the precedence order induced by the game tree
    \\ $h \land h'$ & the lowest common ancestor of $h$ and $h'$
    \\ $\Root_i$ & the empty sequence of player $i$ 
    \\[5mm] \header{Correlation and Correlated Profiles}
       $\tau$ & a trigger
    \\ $\bot$ & the empty trigger
    \\ $\bar\tau$ & where $\tau$ was activated---that is, $I$ if $\tau = Ia$ is a sequence, and $\tau$ otherwise.
    \\ $(h, a, \tau)$ & a node in the augmented game 
    \\ $h^\tau$ & shorthand for $(h, \bot, \tau)$ for NFCCE and EFCE; or $(h, *, \tau)$ for EFCCE
    \\ $\vec\xi \in \Xi^c$ & a correlation plan, in notion $c$
        \\ $P \in \mc P$ & a public state
    \\ $R^c_\mediator$ & the size of the representation of $\Xi^c$
    \\[5mm] \header{Parameters}
       $b$ & the (non-chance) branching factor
    \\ $d$ & the depth of the game tree
    \\ $k$ & the information complexity
    \\[5mm] \header{von Stengel--Forges and Column Generation}
       $\Sigma$ & the set of {\em relevant} joint sequences
    \\ $\sigma_1 \bowtie \sigma_2$ & $(\sigma_1, \sigma_2) \in \Sigma$
    \\ $\sigma_1 \bowtie I$ & $(\sigma_1, Ia) \in \Sigma$ for all $a \in A_I$
    \\ $\mc V$ & the von Stengel--Forges polytope
    \\ $\Sr[1]$ & the set of semi-randomized plans of player $i$
    \end{tabular}
    \label{tab:notation}
    \caption{Summary of notation used in the paper}
\end{table}
\end{document}